\documentclass[cmp,numbook,envcountsame]{svjour}

\usepackage{amsmath} 
\usepackage{amssymb}
\usepackage{amsfonts}

\usepackage{cite}

\newcommand{\CC}{{\mathbb C}}
\newcommand{\II}{{\mathbb I}}
\newcommand{\KK}{{\mathbb K}}

\newcommand{\NN}{{\mathbb N}}
\newcommand{\RR}{{\mathbb R}}

\newcommand{\ZZ}{{\mathbb Z}}

\newcommand{\cB}{{\mathcal{B}}}
\newcommand{\cC}{{\mathcal{C}}}
\newcommand{\cD}{{\mathcal{D}}}

\newcommand{\cH}{{\mathcal{H}}}
\newcommand{\cK}{{\mathcal{K}}}

\newcommand{\cR}{{\mathcal{R}}}
\newcommand{\cS}{{\mathcal{S}}}

\newcommand{\bLambda}{{\mbox{\boldmath $\Lambda$}}}

\newcommand{\bia}{\mbox{\boldmath $a$}}
\newcommand{\bib}{\mbox{\boldmath $b$}}
\newcommand{\bix}{\mbox{\boldmath $x$}}
\newcommand{\biy}{\mbox{\boldmath $y$}}
\newcommand{\biP}{\mbox{\boldmath $P$}}
\newcommand{\biQ}{\mbox{\boldmath $Q$}}
\newcommand{\biR}{\mbox{\boldmath $R$}}

\newcommand{\Co}{{\mbox{\rm C}_0}}

\newcommand{\bdot}{\mbox{\boldmath $\cdot$}}

\newcommand{\lap}{\lambda^\prime}
\newcommand{\lapp}{\lambda^{\prime \prime}}
\newcommand{\Lap}{\Lambda^\prime}
\newcommand{\Lapp}{\Lambda^{\prime \prime}}

\newcommand{\lsml}{{\lap \smile \lapp}}
\newcommand{\tlsml}{\mbox{\tiny \boldmath $\lambda^\prime \! \! \smile \! \! \lambda^{\prime \prime}$}}

\newcommand{\lfrl}{{\lap \frown \lapp}}

\newcommand{\Z}{{\ZZ^d}}
\newcommand{\Lml}{{\Lambda \backslash \lambda}}

\newcommand{\Ms}{M^{\mbox{\tiny \boldmath $\lambda^\prime \! \! \smile \! \! \lambda^{\prime \prime}$}}}

\newcommand{\bVl}{\mbox{\boldmath $V$}_\lambda}
\newcommand{\sbVl}{\mbox{\boldmath $\scriptstyle V$}_{\! \lambda}}

\def\eg{{\it e.g.\ }}
\def\ie{{\it i.e.\ }}
\def\viz{{\it viz.\ }}

\title{The resolvent algebra for oscillating lattice systems.    
Dynamics, ground and equilibrium states}

\author{Detlev Buchholz}
\institute{Institut f\"ur Theoretische Physik, 
Universit\"at G\"ottingen, \\ 37077 G\"ottingen, Germany}

\authorrunning{Detlev Buchholz} 
\titlerunning{The resolvent algebra for oscillating lattice systems}

\date{}

\begin{document}

\maketitle

\begin{abstract} 
Within the C*-algebraic framework of the resolvent algebra
for can\-onical quantum systems, the structure of 
oscillating lattice systems with bounded nearest neighbor interactions
is studied in any number of dimensions. The global dynamics of such 
systems acts on the resolvent algebra by automorphisms and there exists 
a (in any regular representation) weakly dense subalgebra 
on which this action is pointwise norm continuous. Based on this observation,  
equilibrium (KMS) states as well as ground states are constructed 
which are shown to be regular.  It is also indicated 
how to deal with singular interactions and non-harmonic oscillations.  
\end{abstract}

\section{Introduction}
\setcounter{equation}{0}

In this article we continue our study of the resolvent algebra
with emphasis on applications. 
The resolvent algebra, being a novel approach to the treatment of
canonical commutation relations of finite and infinite systems, 
was originally invented to overcome some technical 
difficulties in the analysis of supersymmetric models \cite{BuGr1}. 
It was quickly realized that, in contrast to the 
familiar Weyl algebra, it provides
a general C*-algebraic framework for the construction of non-trivial 
dynamics and of the corresponding 
states \cite{BuGr2}. Moreover, it encodes distinctive properties of 
finite and infinite quantum systems and has an intriguing algebraic 
structure~\cite{Bu1}. For a recent review of these developments, 
cf.\ \cite{BuGr3}.

It is the aim of the present article to supplement these results by a 
study of oscillating lattice systems with bounded nearest neighbor 
interaction, describing particles which are confined about 
their respective lattice positions by harmonic forces and which interact with 
each other. We shall show that for such systems the time 
evolution acts by automorphisms on the resolvent algebra.
This fact was already established in \cite{BuGr2} for one-dimensional
lattices and we extend here these results to any number of dimensions.
Moreover, we shall exhibit a subalgebra of the resolvent algebra which is
weakly dense in all regular representations and on which the   
automorphisms act pointwise norm continuously, \ie this subalgebra
together with the respective automorphisms constitutes C*-dynamical
systems. 

This observation facilitates the construction of global equilibrium 
and ground states, a topic which has not been discussed in the W*-dynamical 
approach to oscillating lattice systems, proposed in 
\cite{NaRaSchSi,NaSchSiStZa} and 
references quoted there. There the underlying algebra is equipped with the 
weak topology induced by some {\it ad hoc} choice of states. 
Yet, choosing from the outset a weak topology on the global 
algebra is a quite subtle issue since 
different equilibrium or ground states lead in
general to globally disjoint representations of the algebra 
\cite{Ta} and hence to different topologies. 
The best one may hope for in the present context is that the
representations of interest are quasi-equivalent on all subalgebras
affiliated with finite subsets of the lattice  
(\viz locally normal), but this feature requires some proof.

In order to illustrate the utility of the present approach in this
respect, we will construct global equilibrium (KMS) states for high 
temperatures as well as ground states 
on the above subalgebra for any of the given dynamics and show 
that they are all locally normal. 
They can then be extended to regular equilibrium
and ground states on the full resolvent algebra. 
Results of this type seem to be relatively rare, 
cf.\ for example the books \cite{AlKoKoRo,BrRo}
and references quoted there.
It is the primary purpose of the present article to reveal 
the simplifying features of our algebraic 
approach for the study of lattice systems; but it 
should be noted that the resolvent algebra can also be used for the
treatment of continuous systems (field theories) \cite{BuGr2}.

The subsequent section contains definitions and two relevant 
technical results. In Sec.~3 we establish the existence
of interacting dynamics of the resolvent algebra 
and exhibit the subalgebra which is stable and pointwise norm continuous 
under their action. Sec.~4 contains the construction 
of locally normal 
equilibrium (KMS) states for sufficiently high temperatures and 
Sec.~5 that of ground states.  In Sec.~6 we indicate how  
to deal with 
non-harmonic dynamics and singular interactions.
The article concludes with a brief summary and outlook.

\section{Preliminaries}
\setcounter{equation}{0}

Let $\Z$ be a $d$-dimensional cubic lattice, let 
$\Lambda \subset \Z$ denote any of 
its finite subsets and, more specifically, let 
$\lambda \in \Z$ denote any lattice point. 
For the case at hand, the resolvent algebra $\cR(\Z)$ 
is generated by the sums and products of the resolvents 
\begin{equation} \label{r2.1} 
(i c + \sum_j \bia_j \biP_{\lambda_j} + 
\sum_k \bib_k \biQ_{\lambda_k})^{-1} \, ,
\quad c \in \RR \backslash \{0\} \, , \ \bia_j, \bib_k \in \RR^d  \, ,
\end{equation}
formed by linear combinations of 
$d$-dimensional momentum and position 
operators $(\biP_\lambda, \biQ_\lambda)$ associated with 
the lattice points $\lambda \in \Z$; the bold face 
product denotes scalar products.
Operators at different lattice points commute, whereas 
at a given point they satisfy canonical commutation relations
in resolvent form, cf.~\cite[Def.~3.1]{BuGr2}. 
Similarly, one defines the resolvent algebras
$\cR(\Lambda)$ for finite subsets $\Lambda \subset \Z$
as the algebras generated by resolvents of 
$(\biP_\lambda, \biQ_\lambda)$
with $\lambda \in \Lambda$. As a matter of fact, one proceeds 
from abstract C*-versions of these
algebras. Yet, as has been shown in
\cite[Thm.~4.10]{BuGr2}, these C*-algebras 
are faithfully represented in any regular representation, \eg the 
Fock representation. So we may assume here that we are dealing with 
the concretely represented resolvent algebra in some such representation.

For the construction of interacting dynamics it matters that
the resolvent algebras have a non-trivial 
(two-sided) ideal structure \cite{BuGr2}.
We do not need to dive deeply here into this topic, cf. \cite{Bu1},
and will only make use of the following elementary facts. 
Given $\lambda \in \Lambda$, one can form the operators
$f(\biP_\lambda) g(\biQ_\lambda)$, where $f,g \in \Co(\RR^d)$,
the space of continuous functions which tend to $0$ at infinity. 
The C*-algebra $\cC(\lambda)$ generated by these operators forms
an ideal in the resolvent algebra $\cR(\lambda)$ which is isomorphic
to the algebra of compact operators on some separable Hilbert space;
more generally, given $\Lambda \subset \Z$, the products  
$\Pi_{\lambda \in \Lambda} \, f_\lambda (\biP_\lambda) g_\lambda (\biQ_\lambda)$,
where $f_\lambda, g_\lambda \in \Co(\RR^d)$, $\lambda \in \Lambda$,
form an ideal $\cC(\Lambda) \subset 
\cR(\Lambda)$ which, again, is isomorphic to 
the algebra of compacts \cite[Thm.\ 5.4]{BuGr2}. 
Note that the algebras 
$\cC(\Lambda)$ are isomorphic to the unique 
C*-tensor products $\otimes_{\lambda \in \Lambda} \cC(\lambda)$,
$\Lambda \subset \Z$, since the algebras $\cC(\lambda)$ 
are postliminal (type I) and commute with each other at different points 
$\lambda \in \Z$. 

Whereas the assignment 
$\Lambda \mapsto \cR(\Lambda)$ defines a net of C*-algebras
on $\Z$ so that we 
can proceed to its C*-inductive limit $\cR(\Z)$, the compact algebras
$\cC(\Lambda)$ are not nested. We therefore proceed from them to the 
algebras $\cK(\Lambda)$ which are formed by all compact 
algebras in the respective region $\Lambda$ and the unit operator. 

\medskip
\noindent \textbf{Definition:}  Let $\Lambda \subset \Z$. The 
algebra $\cK(\Lambda) \subset \cR(\Z)$ is the unital C*-algebra formed by
all algebras $\cC(\Lambda_1)$ with $\Lambda_1 \subseteq \Lambda$, 
and the unit operator $1$. 

\medskip
The algebras $\cK(\Lambda)$ are isomorphic to 
the unique tensor product $\otimes_{\lambda \in \Lambda} \cK(\lambda)$, 
and the assignment $\Lambda \mapsto \cK(\Lambda)$ defines a subnet 
of the resolvent algebra whose C*-inductive limit is denoted by $\cK(\Z)$. 
This subalgebra of the resolvent algebra will play a prominent role 
in the subsequent discussion.

Let us turn next to the definition of dynamics. We proceed from the
harmonic dynamics $\alpha^{(0)}_\RR$ which, for given frequency 
$\varpi > 0$ and any time $t \in \RR$, is fixed by its action on the 
position and momentum operators at any $\lambda \in \Z$,  
$$
\alpha_t^{(0)} (\biQ_\lambda) = \cos(t \varpi) \biQ_\lambda +
\sin(t \varpi)/\varpi \, \biP_\lambda,
\ 
\alpha_t^{(0)} (\biP_\lambda) = \cos(t \varpi) \biP_\lambda -
\varpi \sin(t \varpi) \,  \biQ_\lambda \, .
$$
Note that each compact subalgebra $\cC(\Lambda)$,
$\Lambda \subset \Z$, is stable under this action. 

The full dynamics  $\alpha_\RR$ is obtained from
$\alpha^{(0)}_\RR$ by 
introducing between neighboring points in increasing  
subsets $\Lambda \subset \Z$ some interaction. 
Pairs of neighboring points will be 
denoted by  
$\lsml \in \Lambda \times \Lambda$
and the interaction between any two such points is described  
by the (real) potential 
$ V_{\lfrl} \doteq
V(\biQ_{\lap} - \biQ_{\lapp})$,
where $V \in \Co(\RR^d)$ is kept fixed. 
Proceeding to the interaction picture, we consider 
cocycles $\gamma^\Lambda_\RR$ which are designed to define  
automorphisms of the resolvent
algebra. They satisfy the cocycle equation 
$\alpha_s^{(0)} \gamma_t^\Lambda \alpha_{-s}^{(0)} \gamma_s^\Lambda =
\gamma_{s + t}^\Lambda$ for \mbox{$s,t \in \RR$}  with the initial value  
$\partial_t \gamma_t^\Lambda (\, \bdot \,)|_{t = 0} = i \, 
\big[ \, \bdot \, , 
\sum_{\lsml \in \Lambda \times \Lambda} 
\,   V_{\lfrl} \big]$. 
Here the dot stands for elements of the 
resolvent algebra and the square bracket denotes the commutator;
cf.~below for precise definitions. 

We shall prove in the 
subsequent section that there exist solutions to these equations.
With their help one can define automorphisms \ 
$\alpha^\Lambda_t \doteq \alpha_t^{(0)} \gamma_{-t}^\Lambda$, $t \in \RR$, 
of the resolvent algebra, describing the dynamics in the 
presence of interaction between neighboring points in the sets 
$\Lambda \subset \Z$. Moreover, these automorphisms converge 
pointwise in norm to automorphisms  
$\alpha_t$, $t \in \RR$, of the resolvent algebra in the 
thermodynamic limit $\Lambda \nearrow \Z$. The latter automorphisms
no longer leave the individual algebras $\cC(\Lambda)$,
$\Lambda \subset \Z$, invariant, but
the subalgebra $\cK(\Z)$ remains stable under their action,
as we shall see.

For the proof of these assertions we need two technical lemmas which
we supply here for later reference. We begin with a definition.

\medskip 
\noindent \textbf{Definition:} Let 
$\lsml \in \Z \times \Z$ and let 
$$
(\biP_{\lfrl}, 
\biQ_{\lfrl})
\doteq (\biP_{\lap} - \biP_{\lapp},  
\biQ_{\lap} - \biQ_{\lapp}) \, .
$$ 
The C*-algebra generated by the resolvents of linear combinations of 
the components of $(\biP_{\lfrl}, \biQ_{\lfrl})$ 
is denoted by $\cR(\lfrl)$ and its subalgebra 
which is generated by the operators $f(\biP_{\lfrl}) g(\biQ_{\lfrl})$ 
with $f,g \in \Co(\RR^d)$ is denoted by 
$\cC(\lfrl) \subset \cR(\lfrl)$.
(Note that the latter algebra is,
once again, isomorphic to the algebra of 
compact operators on some Hilbert space; but it is not
contained in $\cK(\Z)$.) 

\begin{lemma} \label{l2.1}
Let $\lsml \in \Z \times \Z$ 
and let $ V_{\lfrl}(s) 
\doteq \alpha_s^{(0)}( V_{\lfrl})$,
$s \in \RR$. Then, for any $t_1, t_2 \in \RR$, one has 
$\int_{t_1}^{t_2} \! ds \,  V_{\lfrl}(s) 
\in \cC(\lfrl)$, where the 
integral is defined in the strong operator topology of the chosen 
regular (hence faithful) representation. 
\end{lemma} 
\begin{proof} \ 
In order to simplify the notation, we put 
$\biP \doteq \biP_{\lfrl}$, 
$\biQ \doteq \biQ_{\lfrl}$ 
and note that \
$\alpha_s^{(0)}( V_{\lfrl}(\biQ)) =
V(\cos(\varpi s) \, \biQ + \sin(\varpi s) / \varpi \, \biP)$ 
is contained in $\cR(\lfrl)$, c.f. 
\cite[Prop.\ 5.2]{BuGr2}. By von Neumann's uniqueness theorem, 
the automorphic action of \ $\alpha^{(0)}_s$ on 
$\cR(\lfrl)$ is implemented in any regular
representation of this algebra by unitary operators 
$U_0(s)$ which depend continuously on $s \in \RR$ in the 
strong operator topology. Since the potential~$V$ is bounded,
the above integral is well defined in this topology. 

In order to see that the integral 
is an element of $\cC(\lfrl)$,  
we make use of the fact that the elements of $\Co(\RR^s)$
can be approximated in norm by Schwartz test functions and 
assume temporarily that the potential $V$ belongs 
to this subspace. We then study the (distributional) kernel of 
the operator function \
$s,s^\prime \mapsto  V_{\lfrl} \, U_0(s-s^\prime)
\,  V_{\lfrl} $
in the Schr\"odinger representation. Making use of Dirac's bra-ket 
notation, this kernel is given in position space by 
$$
\langle \bix |  V_{\lfrl} \, 
U_0(s-s^\prime)
\,  V_{\lfrl} | \biy \rangle
=  V(\bix) \, \langle \bix | U_0(s-s^\prime) | \biy \rangle V(\biy) \, ,
\quad \bix, \biy \in \RR^d \, .
$$
The Green's function of the harmonic oscillator, 
$\bix, \biy \mapsto \langle \bix | U_0(s-s^\prime) | \biy \rangle $,
is known to be continuous and bounded for fixed 
$(s-s^\prime) \not\in (\pi / \varpi) \, \ZZ$, \ cf. \cite{BaVeKh}. 
Hence for those
values the above kernel is square integrable 
in $\bix, \biy$ and therefore belongs to a 
Hilbert Schmidt operator. Whence, the operators 
$$ 
\alpha_{s^\prime}^{(0)}( V_{\lfrl}) \, 
\alpha_{s}^{(0)}( V_{\lfrl})
= U_0(s^\prime) \  V_{\lfrl} \, 
U_0(s - s^\prime) \,  V_{\lfrl}  \
U_0(-s) 
$$
are compact and consequently are elements of 
$\cC(\lfrl)$
if $(s^\prime - s) \not\in (\pi / \varpi) \, \ZZ$. 
Moreover, they are uniformly bounded for all $s^\prime,s \in \RR$. 

We proceed by considering  the double integral 
$$
\int_{t_1}^{t_2} \! ds^\prime \int_{t_1}^{t_2} ds \, 
\alpha_{s^\prime}^{(0)}( V_{\lfrl}) \, 
\alpha_{s}^{(0)}( V_{\lfrl}) \, .
$$
Disregarding the singular set $\{ (s^\prime, s) : 
(s^\prime -s) \in (\pi / \varpi) \, \ZZ \}$ which has 
measure zero, the integration extends over a strong-operator 
continuous and bounded function which has values in the 
compact operators. Thus if $\Phi_n$, $n \in \NN$, is any 
sequence of vectors which converges weakly to zero, one 
obtains by the dominated convergence theorem
\begin{align*}
& \lim_n \| \! \int_{t_1}^{t_2} \! ds \, 
\alpha_s^{(0)}( V_{\lfrl}) \, \Phi_n \|^2 \\
& = \lim_n \, 
\int_{t_1}^{t_2} \! ds^\prime \!\int_{t_1}^{t_2} \! ds \, 
\langle \Phi_n, \alpha_{s^\prime}^{(0)}( V_{\lfrl}) \, 
\alpha_{s}^{(0)}( V_{\lfrl}) \, \Phi_n \rangle = 0 \, . 
\end{align*}
Hence  $\int_{t_1}^{t_2} \! ds \, \alpha_s^{(0)}( V_{\lfrl})$
is compact and therefore an element of $\cC(\lfrl)$.
This proves the assertion for potentials $V \in \Co(\RR^d)$
which are test functions;
the statement then follows from the continuity of the integral
with regard to $V$ in the norm topology of $\Co(\RR^d)$. \qed 
\end{proof}

\medskip
While the algebras $\cC(\lfrl)$ are not contained
in the C*-algebra $\cK(\Z)$, their elements induce bounded 
derivations of it; note that 
$\cK(\Z)$ is not a simple algebra. The proof of this fact,  
entering into our construction of dynamics, is based on the 
following lemma. 

\begin{lemma} \label{l2.2}
Let $\Lambda \subset \Z$ and let 
$\, \lsml \in \Z \times \Z \, $
be nearest neighbors. If \ $\lap \in \Lambda$, 
$\lapp \not\in \Lambda$ one has 
$\cC(\lfrl) \, 
\cC(\Lambda), \ \cC(\Lambda) \, \cC(\lfrl) 
\subset \cC(\Lambda \cup \lapp)$.
Similarly, if $\lap, \lapp \in \Lambda$
one has $\cC(\lfrl) \,
\cC(\Lambda), \ \cC(\Lambda) \, \cC(\lfrl) 
\subset \cC(\Lambda)$. 
(Note that 
$\cC(\lfrl) =
\cC(\lapp \frown \lambda^{\prime})$.)
\end{lemma}
\begin{proof}
As to the first part of the statement, it suffices because of the  
tensor product structure of $\cC(\Lambda)$ to consider the simplest 
non-trivial case, where $\Lambda = \lap$.
Now the linear space 
$\cC(\lfrl) \, \cC(\lap)$ 
is generated by the operators 
$f(\biP_{\lfrl}) 
g(\biQ_{\lfrl}) \, 
h(\biP_{\lap}) k(\biQ_{\lap}) $,
where $f,g,h,k \in \Co(\RR^d)$. So we must show that these 
operators, lying in the algebra $\cR(\lap \cup \lapp)$, 
actually belong to its compact ideal. 

This task can be accomplished by explicit computations in the 
Schr\"odinger representation of the algebra 
$\cR(\lap \cup \lapp)$, where one has to 
show that the operators are compact. 
Alternatively, one arrives at this conclusion by 
noticing that the operators lie in the principle ideal 
which is generated by the $4d$-fold product of resolvents 
involving all components of the operators 
$\biP_{\lfrl}, 
\biQ_{\lfrl},  
\biP_{\lap}, \biQ_{\lap}$ in 
the given order, cf. \cite[Prop.~3.8]{BuGr2}.
This ideal, being equal to the intersection of 
the principle ideals formed by the individual resolvents 
\cite[Prop.~3.8]{BuGr2}, 
does not depend on this order, so 
it coincides with the principle ideal generated by 
the $4d$-fold product of resolvents involving the 
components of the operators 
$\biP_{\lfrl}, 
\biP_{\lap},
\biQ_{\lfrl},  
\biQ_{\lap}$.  
Since the product of the resolvents of the 
components of $\biP_{\lfrl}, 
\biP_{\lap}$ is of the form $f_2(\biP_{\lap}, \biP_{\lapp})$
with $f_2 \in \Co(\RR^d) \otimes \Co(\RR^d)$, 
and similarly for $\biQ_{\lfrl}, 
\biQ_{\lap}$, cf. \cite[Prop.~5.2]{BuGr2}, 
this ideal is generated by 
the operators $f_2(\biP_{\lap}, \biP_{\lapp})
g_2(\biQ_{\lap}, \biQ_{\lapp})$,
where $f_2, g_2 \in \Co(\RR^d) \otimes \Co(\RR^d) = \Co(\RR^{2d})$.
It therefore coincides with the compact ideal 
$\cC(\lap \cup \lapp)$,  
which consequently contains 
$\cC(\lfrl) \cC(\lap)$.
Similarly, one obtains  
$\cC(\lap) \, \cC(\lfrl)
\subset \cC(\lap \cup \lapp)$, 
completing the proof of the first part of the statement.
The second part immediately follows from the fact that 
$\cC(\lfrl) \subset \cR(\lap \cup \lapp)$
and  that the subalgebra $\cC(\lap \cup \lapp) \subset 
\cR(\lap \cup \lapp)$ is an ideal.  \qed
\end{proof}

\section{Dynamics}
\setcounter{equation}{0}

In this section we establish the existence  
of one-parameter groups of automorphisms of the resolvent algebra 
(dynamics) which describe nearest neighbor interactions. These results extend
those obtained in \cite{BuGr2} for lattice systems in $d = 1$ dimension. 
We could partially rely on results established 
subsequently in \cite{NaRaSchSi,NaSchSiStZa} 
and references quoted there; but we need some more detailed information 
and therefore present here a different, simpler argument. 

As has been explained in the preceding section, we must show that 
for any finite $\Lambda \subset \Z$ there exist cocycles 
$\gamma^\Lambda_\RR$ which act as automorphisms on the resolvent 
algebra $\cR(\Z)$ and satisfy 
\begin{equation} \label{r3.1} 
\alpha_s^{(0)} \gamma^\Lambda_t \alpha_{-s}^{(0)} \, \gamma^\Lambda_s
= \gamma^\Lambda_{s + t} \quad \text{for} \quad s,t \in \RR 
\end{equation}
with initial condition given by the derivative (on a suitable domain)
\begin{equation} \label{r3.2} 
\partial_t \gamma^\Lambda_t(\, \bdot \,) \, |_{t = 0} = 
i \, [ \, \bdot \, , V_\Lambda]
\, , \quad \text{where} \quad V_\Lambda \doteq 
\sum_{\lsml \in \Lambda \times \Lambda} V_\lfrl \, .
\end{equation}
Here $\alpha^{(0)}_\RR$ denotes the harmonic dynamics acting 
at each lattice site and the dot stands for any element 
of the resolvent algebra. We will exhibit solutions to this problem
by proceeding to its integrated version,  
\begin{equation}
\gamma_t^\Lambda(\, \bdot \,) = \iota(\, \bdot \,) \, + \, i \! \int_0^t \! ds \,
[\gamma_s^\Lambda (\, \bdot \,) ,
V_\Lambda(s)] \, , \quad \text{where} \quad V_\Lambda(s) \doteq 
\alpha_s^{(0)}(V_\Lambda) 
\end{equation}
and $\iota$ denotes the identity map. These solutions are 
obtained by iteration, giving the Dyson series 
\begin{equation} \label{r3.4}
\gamma_t^\Lambda(\, \bdot \,) =
\iota(\, \bdot \,) \, + \, \sum_{n = 1}^\infty 
i^n \int_0^t \! ds_n \int_0^{s_{n}} \! ds_{n-1}
\dots \int_0^{s_{2}} \! ds_1 \, 
[ \dots [ \ \bdot \ , V_\Lambda(s_1) \, ],  \dots ,V_\Lambda(s_n) \, ] \, .
\end{equation}
In this form the cocycles $\gamma^\Lambda_\RR$ 
acquire some precise mathematical meaning:
The functions $s \mapsto V_\Lambda(s)$ are continuous in the strong
operator topology of the chosen (regular) representation, so the
integrals are well defined in this topology. Moreover, since 
$\| V_\Lambda(s) \| = \| V_\Lambda \| < \infty$, the series is 
absolutely convergent in norm, so the cocycles $\gamma^\Lambda_t$ map the 
elements of the resolvent
algebra into bounded operators. The more difficult part 
is the demonstration that the images of the resolvent algebra 
under this action are actually contained in the resolvent algebra
itself and depend pointwise norm 
continuously on $t \in \RR$. Since the underlying
representation is faithful, one then has shown that these 
properties are representation independent.

To accomplish this goal we note that it suffices to establish these 
properties for the individual terms appearing in the Dyson series
because of its convergence properties. Moreover, these 
terms can be split further into finite sums of linear maps 
$\Ms_n(t)$, $t \in \RR$, 
mapping the resolvent algebra $\cR(\Z)$ 
into bounded operators. They are given by 
\begin{align} \label{r3.5}
& \Ms_n(t)(\, \bdot \, ) \nonumber \\ 
& \doteq
\int_0^t \! ds_n \int_0^{s_n} \! ds_{n-1}
\dots \int_0^{s_2} \! ds_1 \, 
[ \dots [ \ \bdot \ , V_{\lap_1 \frown \lapp_1}(s_1) \, ],  
\dots ,V_{\lap_n \frown \lapp_n}(s_n) \, ] \, ,
\end{align}
where 
$\mbox{\boldmath $\lap \! \! \smile \! \! \lapp$} 
\doteq \lap_1 \smile \lapp_1 , \dots , \lap_n \smile \lapp_n \in 
\Lambda \times \Lambda$ are $n$ pairs of nearest neighbors. 
In the subsequent lemma we collect some pertinent properties of 
these maps.

\begin{lemma} \label{l3.1}
Let $\Lambda \subset \Z$ and let 
$\mbox{\boldmath $\lap \! \! \smile \! \! \lapp$} \doteq 
\lap_1 \smile \lapp_1 , \dots , \lap_n \smile \lapp_n \in 
\Lambda \times \Lambda$, $n \in \NN$, 
be any collection of nearest neighbors
(which need not be different from each other).  
The maps $\Ms_n(t)$, $t \in \RR$, defined in 
\eqref{r3.5}, have the properties
\begin{itemize}
\item[(i)]  \ $\| \Ms_n(t) - \Ms_n(s) \| \leq  2^n \, \| V \|^n \, 
|t^n - s^n| / n ! \, $ if $s,t \geq 0$ or $s,t \leq 0$, and $\Ms_n(0) = 0$.  
\item[(ii)] \ $\Ms_n(t)$ maps the resolvent algebra $\cR(\Z)$ into itself
\item[(iii)] \ Given $\Lambda_0 \subset \Z$, there is a 
finite $\Lambda_n \supseteq \Lambda_0$, depending on the collection of 
points $\lap_1 \smile \lapp_1 , \dots , \lap_n \smile \lapp_n \in 
\Lambda \times \Lambda$,  such that 
$\Ms_n(t)(\cC(\Lambda_0)) \subseteq \cC(\Lambda_n)$. 
\end{itemize}
\end{lemma}
\begin{proof}
(i) Let $R \in \cR(\Z)$ and let $\Phi$ be any vector 
in the underlying representation space of the
resolvent algebra. Then one obtains for $t  \geq s \geq 0$ 
\begin{align*}
& \|( \Ms_n(t)(R) -  \Ms_n(s)(R) ) \, \Phi \|  \\
& = \| \int_s^t \! ds_n \int_0^{s_n} \! ds_{n-1}
\dots \int_0^{s_2} \! ds_1 \, 
[ \dots [R , V_{\lap_1 \frown \lapp_1}(s_1) \, ],  
\dots ,V_{\lap_n \frown \lapp_n}(s_n) \, ] \, \Phi \, \| \\
& \leq \int_s^{t} \! ds_n \int_0^{s_n} \! ds_{n-1} 
\dots \int_0^{s_2} \! ds_1 \, \|
[ \dots [R , V_{\lap_1 \frown \lapp_1}(s_1) \, ],  
\dots ,V_{\lap_n \frown \lapp_n}(s_n) \, ] \, \Phi \, \| \\[1mm]
& \leq 2^n \, \| V \|^n \| R \| \, \| \Phi \| \,
(t^n - s^n)/ n! \, ,
\end{align*}
where, in the last step, 
we estimated the norm of the $n$-fold commutator, making 
use of the fact that $\| V_{\lfrl}(s) \| = \| V \|$
does neither depend on time nor on the chosen 
family of neighboring points. 
The first part of the statement then follows for the particular choice of $t,s$
and the other cases are treated in a similar manner. The
equality $\Ms_n(0) = 0$ is immediate from relation \eqref{r3.5}.

(ii) The second statement is proved by induction in $n \in \NN$,
so let $R \in \cR(\Z)$. For $n = 1$ we have 
$\Ms_1(t)(R) = \int_0^t \! ds \, [R, V_{\lap_1 \frown \lapp_1}(s)]$;
this is an element of $\cR(\Z)$ since 
$ \int_0^t \! ds \, V_{\lap_1 \frown \lapp_1}(s) \in \cC(\lap_1 \frown \lapp_1)
\subset \cR(\Z)$ according to Lemma~\ref{l2.1}.
For the induction step from $n$ to $n + 1$ we notice that 
$$
\Ms_{n+1}(t)(R) =  
\int_0^t \! ds \, [M_{n}^\Lambda(s)(R), V_{\lap_{n+1} \frown \lapp_{n+1}}(s)]
$$ 
and the function $s \mapsto M_{n}^\Lambda(s)(R) \in \cR(\Z)$ is continuous in 
norm, as has been shown in part (i). We approximate the 
integral by sums of the form
$$
S_N(R) \doteq \sum_{k=0}^N \int_{t_k}^{t_{k + 1}} \! ds \, [\Ms_n(t_k)(R),
V_{\lap_{n+1} \frown \lapp_{n+1}}(s)] \, ,
$$
where all distances $| t_k - t_{k+1} |$, $k = 0, \dots, N$, tend to zero
in the limit of large $N$. Applying Lemma \ref{l2.1} again, one sees
that each term in these sums is an element of the resolvent algebra, hence 
$S_N(R) \in \cR(\Z)$. 
An estimate as in part (i) yields
$$
\| \Ms_{n+1}(t)(R) - S_N(R) \|
\leq 2^{n+1} \| V \|^{n+1} \, 
\sum_{k=0}^N \int_{t_k}^{t_{k + 1}} \! ds \, 
| t_k^n - s^n|/ n! \ ,
$$
showing that $ \Ms_{n+1}(t)(R)$ can be approximated in norm by the 
sums $S_N(R)$. Since $\cR(\Z)$ is closed in this topology, 
the second statement follows.

(iii) For the third statement we proceed as in part (ii)
but need to have a closer look at the localization properties  of
the operators. We make use again of an induction argument, so 
let $\Co \in \cC(\Lambda_0)$. For $n = 1$ we get 
$$
\Ms_1(t)(\Co) = \int_0^t \! ds \, [\Co, V_{\lap_1 \frown \lapp_1}(s)]
=  [\Co, \int_0^t \! ds \, V_{\lap_1 \frown \lapp_1}(s)  ] \, ,
$$
where $\int_0^t \! ds \, V_{\lap_1 \frown \lapp_1}(s) \in 
\cC(\lap_1 \frown \lapp_1)$ according to Lemma \ref{l2.1}. 
If both points $\lap_1, \lapp_1 \not\in \Lambda_0$, then the 
commutator vanishes because of the commutativity of operators at
different lattice points. If $\lap_1 \in \Lambda_0$, 
$\lapp_1 \not\in \Lambda_0$, then we have 
$$
[\Co, \int_0^t \! ds \, V_{\lap_1 \frown \lapp_1}(s)  ] \in  
\big( \cC(\Lambda_0) \cC(\lap_1 \frown \lapp_1) - \cC(\lap_1 \frown \lapp_1) 
\cC(\Lambda_0) \big) \subset \cC(\Lambda_0 \cup \lapp_1) \, ,
$$ 
where the second inclusion follows from Lemma \ref{l2.2}; in a 
similar manner, one treats the case  $\lapp_1 \in \Lambda_0$, 
$\lap_1 \not\in \Lambda_0$. Finally, if $\lap_1, \lapp_1 \in \Lambda_0$,
then the commutator is an element of $\cC(\Lambda_0)$ according to 
Lemma \ref{l2.2}, completing the proof of the initial step of the induction. 
Next, let $\Ms_n(t)(\Co) \in \cC(\Lambda_n)$, where 
$\Lambda_n \supseteq \Lambda_0$ is some finite set. 
For the induction step from $n$ to $n +1$ we approximate 
the integral $\Ms_{n+1}(t)(\Co) = \int_0^t \! ds \, 
[\Ms_n(s)(\Co), V_{\lap_{n+1} \frown \lapp_{n+1}}(s) \, ]$
by sums which converge in the norm topology, cf.\ the preceding part (ii).  
For the terms in these sums we have the inclusion
\begin{align*}
& [\Ms_n(t_k)(\Co),
\int_{t_k}^{t_{k + 1}} \! \! ds \, V_{\lap_{n+1} \frown \lapp_{n+1}}(s)] \\
& \in \ \big(\cC(\Lambda_n) \, \cC(\lap_{n+1} \frown \lapp_{n+1})   
-   \cC(\lap_{n+1} \frown \lapp_{n+1}) \, \cC(\Lambda_n) \big)
\end{align*}
if $\lap_{n+1} \in \Lambda_n$ or $\lapp_{n+1} \in \Lambda_n$; in all other
cases the commutator vanishes. By another application of Lemma \ref{l2.2}
we conclude that these operators are elements of $\cC(\Lambda_{n+1})$,
where $\Lambda_{n+1} \supseteq \Lambda_n$. Since the latter algebra is 
closed in the norm topology, this completes the proof of the lemma. \qed
\end{proof}

Having control on the individual terms appearing in the Dyson series, 
we now need to study their behavior in the thermodynamic
limit $\Lambda \nearrow \Z$. To this end we combine them into 
$n$-th order contributions to the Dyson series, $n \in \NN$, 
\begin{equation} \label{r3.6}
D_n^\Lambda(t) 
\doteq 
\sum_{\mbox{\tiny \boldmath $\lap \! \! \smile \! \! \lapp$} =
\lap_1 \smile \lapp_1, \dots , \lap_n \smile \lapp_n 
\in \Lambda \times \Lambda }  \, 
\Ms_n(t)  \, .
\end{equation}
We then have the following result for these maps.

\begin{lemma} \label{l3.2}
Let $\Lambda_0 \subset \Z$ be a finite set and let $n \in \NN$. 
There exists a finite set $\Lambda_n \supset \Lambda_0$ such that
the maps $D_n^\Lambda(t)$, $t \in \RR$, defined above, satisfy 
\begin{itemize}
\item[(i)] the restricted maps $D_n^\Lambda(t) \upharpoonright \cR(\Lambda_0)$
do not depend on $\Lambda$ if   $\Lambda \supseteq \Lambda_n$
\item[(ii)] $\| (D_n^\Lambda(t) - D_n^\Lambda(s)) 
\upharpoonright \cR(\Lambda_0) \| 
\leq 2^{(d + 2) n} \, L_0 \cdots (L_0 + n -1) \,  
\| V \|^n \, |t^n - s^n|  / n! \, $ for $s,t \geq 0$ or 
$s,t \leq 0$, where 
$L_0$ is the number of points in $\Lambda_0$.  
This bound is independent of $\Lambda$.
\end{itemize}
\end{lemma}
\begin{proof}
(i) Let $R_0 \in \cR(\Lambda_0)$. 
Because of the commutativity of operators localized at different 
points of the lattice and the stability of the corresponding algebras
under the action of~$\alpha^{(0)}_\RR$, there contribute to
$\Ms_n(t)(R_0)$ only potentials $V_\lfrl$ which 
are attributed to points $\lap, \lapp$ having a lattice distance
of at most $n$ from some point in $\Lambda_0$. Phrased differently,
$\Ms_n(t)(R_0)$ is different from zero only if the
corresponding points {\boldmath $\lap \! \! \smile \! \! \lapp$}
\textit{all} lie in the region $\Lambda_n \times \Lambda_n$, 
where $\Lambda_n$ is obtained from $\Lambda_0$
by surrounding every point in $\Lambda_0$ with a cube of side length $2n$. 
Since for points 
$\mbox{\boldmath $\lap \! \! \smile \! \! \lapp$} 
\not\in \Lambda_n \times \Lambda_n$ the corresponding
terms in relation~\eqref{r3.6} vanish, this proves the first part of
the statement. 

(ii) For the proof of the second part we estimate the 
number $N_n(R_0)$ of terms which contribute to 
$D_n^\Lambda(t)(R_0)$, $t \in \RR$, in 
the sum \eqref{r3.6} and the number of points 
$L_n(R_0)$ in their respective localization regions; 
the norm of the individual terms has 
already been estimated in the preceding lemma. 
Once again,
we rely on induction in $n \in \NN$. For the initial case
$n=1$ we note that, either, one must have $\lap \in \Lambda_0$, 
and there are then $2^{\, d}$ different possibilities for its neighboring 
point $\lapp$, so there appear at most $2^{\, d} L_0$ terms of this 
type in the sum. Interchanging the role of $\lap$, $\lapp$ one gets 
the same number, giving the estimate $N_1(R_0) \leq 2^{\, d + 1} L_0$;
moreover, the number of points $L_1$ in the 
localization region of each term increases at most by 
one,  $L_1 \leq L_0 + 1$. 
As to the induction hypothesis,  we have 
$N_n(R_0) \leq 2^{(d+1)n} \, L_0 \cdots (L_0 + n - 1) $
and $L_n \leq (L_0 + n)$. For the step from $n$
to $n + 1$ we note that we must add to each of the $N_n(R_0)$
contributing configurations 
$\mbox{\boldmath $\lap \! \! \smile \! \! \lapp$}$
of size at most $L_n(R_0)$ another pair $\lap_{n+1} \smile \lapp_{n+1}$
such that the resulting configurations still contribute to the 
corresponding sum \eqref{r3.6}. 
For a given configuration this can be done in at most 
$2^{\, d + 1} L_n(R_0)$ different ways and since there are 
$N_n$ such configurations, we get 
$N_{n+1}(R_0) \leq 2^{\, d + 1} L_n(R_0) N_n(R_0)$ and
$L_{n+1}(R_0) \leq (L_{n}(R_0) +1)$. 
Plugging into these bounds the induction hypothesis, the given estimate
arises. Finally, applying Lemma \ref{l3.1}(i), one gets
\begin{align*}
\| (D_n^\Lambda(t) - D_n^\Lambda(s))(R_0) \| & \leq N_n(R_0)  \sup_{\tlsml}  
\| (\Ms_n(t) - \Ms_n(s))(R_0) \| \\
& \leq 2^{(d + 2) n} L_0 \cdots (L_0 + n -1) \, \| V \|^n \, 
|t^n - s^n| \| R_0 \| / n  \, ,
\end{align*}
completing the proof of the second statement. \qed
\end{proof}

Equipped with this information we can establish the following 
basic result concerning the cocycles.  

\begin{lemma} \label{l3.3}
Let $\Lambda \subset \Z$ and let 
$\gamma^\Lambda_t$, $t \in \RR$, be the map defined by the Dyson 
series~\eqref{r3.4}.
\begin{itemize}
\item[(i)]  The Dyson series for $\gamma^\Lambda_t$ converges absolutely
in norm, uniformly on compact subsets of \ $t\in\RR$.  
It defines automorphisms of the resolvent algebra $\cR(\Z)$ 
and satisfies the cocycle relation \eqref{r3.1}.
\item[(ii)] The function $t \mapsto \gamma^\Lambda_t$ is absolutely 
continuous, uniformly on compact sets of $t \in \RR$.
\item[(iii)] The restrictions 
$\gamma^\Lambda_t \upharpoonright \cK(\Z)$ are 
automorphisms  of $ \cK(\Z)$, $t \in \RR$.
\item[(iv)] The thermodynamic limit 
$\gamma_t \doteq \lim_{\Lambda \nearrow \Z} \gamma^\Lambda_t$  exists 
pointwise on $\cR(\Z)$ in the norm topology, uniformly on 
compact subsets of $t \in \RR$. The limits $\gamma_t$ 
are automorphisms of $\cR(\Z)$ which satisfy 
the cocycle relation \eqref{r3.1},  and $t \mapsto \gamma_t$
is pointwise norm continuous, $t \in \RR$. 
\end{itemize}
\end{lemma}
\begin{proof}
(i) According to relation \eqref{r3.4} we have 
$\gamma^\Lambda_t = \iota + \sum_{n=1}^\infty i^n D^\Lambda_n(t)$, $t \in \RR$,
where the terms $D^\Lambda_n(t)$ were defined in \eqref{r3.6}. 
In the sum in \eqref{r3.6} there appear at most 
$2^{n(d+1)}  L^n$ terms, where $L$ is the number of points in $\Lambda$.
Hence, putting $s=0$ in  Lemma \ref{l3.1}(i), we get
$$
\|   D^\Lambda_n(t) \| \leq 2^{n(d+1)}  L^n \, \sup_{\tlsml}
\| \Ms_n(t) \| \leq  2^{n(d+2)} \, L^n \, \| V \|^n \, |t|^n  / n! \, .
$$
Thus the series of these norms converges uniformly for compact subsets 
of $t \in \RR$. Since the individual terms in the Dyson series map 
the (norm-closed) 
resolvent algebra into itself, cf.\ Lemma \ref{l3.1}(ii),
it follows that $\gamma^\Lambda_t(\cR(\Z)) \subseteq \cR(\Z)$,
$t \in \RR$. That these maps actually are automorphisms of the
resolvent algebra which satisfy the cocycle relation \eqref{r3.1}
can be established by direct computation, based on 
relation \eqref{r3.4}. More conveniently, notice 
that in any regular (hence faithful) representation of $\cR(\Z)$ the 
automorphisms $\gamma^\Lambda_t$ are implemented by the 
adjoint action of unitaries 
$\Gamma^\Lambda_t = e^{it H^{(0)}_\Lambda} e^{-it H_\Lambda}$, $t \in \RR$,
where 
$$
H^{(0)}_\Lambda = 
(1/2) \sum_{\lambda \in \Lambda} (P_\lambda^2 + \varpi^2 Q_\Lambda^2) 
\quad \text{and} \quad
H_\Lambda = H^{(0)}_\Lambda + \sum_{\lsml \in \Lambda \times \Lambda}
V(Q_{\lap} - Q_{\lapp})
$$ 
are selfadjoint operators. Since 
$\gamma^\Lambda_t \doteq \text{Ad} \, \Gamma^\Lambda_t$, $t \in \RR$, 
satisfies the equations
\eqref{r3.1} and \eqref{r3.2} and can be expanded into the
Dyson series \eqref{r3.4}, this completes the proof of 
the first statement.

(ii) As in the preceding step one obtains with the help of Lemma 
\ref{l3.1}(i) the estimate for $s,t \geq 0$ or $s,t \leq 0$
\begin{align*}
\|   D^\Lambda_n(t) -  D^\Lambda_n(s) \| & \leq 2^{n(d+1)}  L^n 
\sup_{\tlsml} \| \Ms_n(t) - \Ms_n(s) \|  \\ 
& \leq  2^{n(d+2)}  L^n \, \| V \|^n \, |t^n - s^n| / n! \, .
\end{align*}
Since the first terms $\iota$ in the two Dyson series cancel, this  
yields 
$$
\| \gamma^\Lambda_t - \gamma^\Lambda_s \| \leq
2^{\, d+2} L \, \| V \| \, |t - s| \, e^{\, 2^{\, d + 2} L \| V \| \, |t|} \, 
\quad \text{for} \quad t \geq s \geq 0 \quad \text{or} \quad 
t \leq s \leq 0 \, ,
$$ 
proving absolute continuity, uniformly on compact sets of $\RR$. 

(iii) Given $\Lambda_0 \subset \Z$ and $n \in \NN$, 
Lemma \ref{l3.1}(iii) implies that there is some 
$\Lambda_n \supseteq \Lambda_0$ such that 
$\Ms_n(t)(\cC(\Lambda_0)) \subset \cC(\Lambda_n)$, $t \in \RR$. 
Hence, because of their continuity 
properties, the terms in the Dyson series map the 
norm closed subalgebra $\cK(\Z) \subset \cR(\Z)$
into itself. Since the Dyson series converges absolutely, the
assertion follows. 

(iv) Let $\Lambda_0 \subset \Z$ and let $R_0 \in \cR(\Lambda_0)$. 
In order to see that 
$\| \gamma^\Lambda_t(R_0) - \gamma^{\Lap}_t(R_0) \|$ tends to
zero for $\Lambda, \Lap \nearrow \Z$, we expand both 
cocycles into a Dyson series. According to Lemma~\ref{l3.2}(i) 
there exists for each $n \in \NN$
some region $\Lambda_n \supset \Lambda_0$ such that for 
$\Lambda, \Lap \supset \Lambda_n$ the first $n$
terms in the Dyson series for $\gamma^\Lambda_t(R_0)$ and
$\gamma^{\Lap}_t(R_0)$ coincide. Hence for such 
$\Lambda, \Lap$ one gets for sufficiently  small $|t|$
\begin{align*}
\| \gamma^\Lambda_t(R_0) - \gamma^{\Lap}_t(R_0) \|
& \leq \sum_{k = n + 1}^\infty \| D^\Lambda_k(t)(R_0) 
-  D^{\Lap}_k(t)(R_0) \| \\
& \leq \| R_0 \| \sum_{k = n + 1}^\infty (\| D^\Lambda_k (t)\|
+ \| D^{\Lap}_k(t) \| )  \\
& \leq 2 \, \| R_0 \| \sum_{k = n + 1}^\infty  
2^{(d + 2) k} \, L_0 \cdots (L_0 + k -1) \,  
\| V \|^k \, |t|^k  / k! \, ,
\end{align*}
where  $L_0$ is the number of points in $\Lambda_0$
and, in the last inequality, we made use of Lemma \ref{l3.2}(ii). 
For times $t$ satisfying $|t| \, 2^{d + 2} \| V \| < 1$ 
this upper bound tends to zero in the limit of large $n$, proving the
norm convergence of $\lim_{\Lambda \nearrow \Z} \gamma^\Lambda_t(R_0)$, uniformly for 
small times. Moreover, since the $\gamma^\Lambda_t$ are automorphisms
of $\cR(\Z)$, which is the C*-inductive limit of the algebras  
$\cR(\Lambda_0)$, $\Lambda_0 \subset \Z$,
this result extends to all elements of $\cR(\Z)$. Thus
the limits $\gamma_t \doteq \lim_{\Lambda \nearrow \Z} \gamma^\Lambda_t$
define automorphisms of $\cR(\Z)$ for small $t \in \RR$.

In order to show that the limit exists for arbitrary  $t \in \RR$, 
we make use of the cocycle relation \eqref{r3.1} for the 
approximating maps, giving for any $R \in \cR(\Z)$ and small
$s,t$ as above
\begin{align*}
\| (\alpha^{(0)}_s \gamma_t \, \alpha^{(0)}_{-s} \gamma_s
- \gamma^\Lambda_{s + t})(R) \|
& = \| (\alpha^{(0)}_s \gamma_t \, \alpha^{(0)}_{-s} \gamma_s -
\alpha^{(0)}_s \gamma^\Lambda_t \alpha^{(0)}_{-s} \gamma^\Lambda_s) (R) \| \\
& \leq \| (\gamma_t - \gamma^\Lambda_t)(\alpha^{(0)}_{-s} \gamma_s(R)) \|
+   \| (\gamma_s - \gamma^\Lambda_s)(R) \| \, .
\end{align*}
Hence $\gamma_{s + t} \doteq \lim_{\Lambda \nearrow \Z} \gamma^\Lambda_{s + t}$ 
exists pointwise on $\cR(\Z)$ in 
norm and coincides with $\alpha^{(0)}_s \gamma_t \, \alpha^{(0)}_{-s} \gamma_s$
for the restricted set of $s,t$. Repeating this procedure, one 
obtains convergence and the cocycle property 
\eqref{r3.1} on all of $\RR$. 

It remains to prove the pointwise norm continuity of $t \mapsto \gamma_t$.
Because of the cocycle property of $\gamma_t$, $t \in \RR$, it suffices to
do this for small $|t|$. Given $\Lambda_0 \subset \Z $
and any $R_0 \in \cR(\Lambda_0)$ we get for small
$s,t \geq 0$ or $s,t \leq 0$ by the same reasoning as above 
the estimate
\begin{align*}
 \| \gamma_t(R_0) - \gamma_s(R_0)  \| 
& \leq \lim_{\Lambda \nearrow \Z} 
\| \gamma^\Lambda_t(R_0) - \gamma^{\Lambda}_s(R_0) \| \\
& \leq  \limsup_{\Lambda \nearrow \Z}  
\sum_{k = 1}^\infty \| D^\Lambda_k(t)(R_0) 
-  D^{\Lambda}_k(s)(R_0) \| \\
& \leq 2 \, \| R_0 \| \sum_{k = 1}^\infty  
2^{(d + 2) k} \, L_0 \cdots (L_0 + k -1) \,  
\| V \|^k \, |t^k - s^k|  / k! \, .
\end{align*}
Since $\Lambda_0$ was arbitrary, the claimed continuity 
obtains on a norm dense subalgebra of $\cR(\Z)$. But 
$\gamma_t$, $t \in \RR$, are automorphisms, so the pointwise norm 
continuity is obtained on all of  $\cR(\Z)$ by standard arguments, 
completing the proof of the lemma. \qed
\end{proof}

Recalling that the dynamics for finite $\Lambda \subset \Z$
is given by $\alpha^\Lambda_t \doteq \alpha^{(0)}_t \gamma^\Lambda_{-t}$,
where $\gamma^\Lambda_t$ has been defined in \eqref{r3.4}, 
$t \in \RR$, we can now state the main result of this section.

\begin{proposition} \label{p3.4}
Let $V \in \Co(\RR^d)$ be any potential and let $\Lambda \subset \Z$.
There exists a corresponding group of automorphisms
$\alpha^\Lambda_\RR$ of $\cR(\Z)$, describing harmonic oscillations 
of the lattice 
system as well as nearest neighbor interactions within the region 
$\Lambda$ for the given potential. 
\begin{itemize}
\item[(i)] The restriction of $\alpha^\Lambda_\RR$ 
to $\cK(\Z) \subset \cR(\Z)$ defines automorphisms
of this subalgebra, and $t \mapsto \alpha^\Lambda_t$, $t \in \RR$, 
acts pointwise norm continuously on it, \ie
$(\cK(\Z),\alpha^\Lambda_\RR)$ is a C*-dynamical system.  
\item[(ii)] The thermodynamic limit $\alpha_t \doteq \lim_{\Lambda \nearrow \Z} 
 \alpha^\Lambda_t$ exists pointwise on $\cR(\Z)$ in the norm topology,
$t \in \RR$, and defines a group of automorphisms 
$\alpha_\RR$ of $\cR(\Z)$ with nearest neighbor interactions  
all over $\Z$. Its restriction $\alpha_\RR \upharpoonright \cC(\Z)$
leaves $\cC(\Z)$ invariant and acts pointwise norm continuously
on this algebra, \ie $(\cK(\Z), \alpha_\RR)$ is a C*-dynamical 
system, as well. 
\end{itemize} 
(Note that the algebra $\cK(\Z)$ is not norm-dense in $\cR(\Z)$; but 
it is dense in 
the strong operator topology in all regular representations, cf. 
\cite[Thm.~5.4]{BuGr2}.)
\end{proposition}
\begin{proof}
It has been shown in Lemma \ref{l3.3}(i) that 
the maps $\gamma^\Lambda_t$ are automorphisms of $\cR(\Z)$ which satisfy 
the cocycle relation \eqref{r3.1}, $t \in \RR$. Since 
$\alpha^{(0)}_\RR$ is
a group of automorphisms of $\cR(\Z)$, this is true
for $\alpha^\Lambda_t \doteq \alpha^{(0)}_t \, \gamma^\Lambda_{-t}$, $t \in \RR$,
as well. In fact, making use of the cocycle relation
for $\gamma^\Lambda_\RR$, one obtains
\begin{align*}
& \alpha^\Lambda_s \, \alpha^\Lambda_t = 
\alpha^{(0)}_s  \gamma^\Lambda_{-s} \ \, \alpha^{(0)}_t  \gamma^\Lambda_{-t}
= \alpha^{(0)}_{s + t} \  \, \alpha^{(0)}_{-t} 
\gamma^\Lambda_{-s} \alpha^{(0)}_t \gamma^\Lambda_{-t} \\
& = \alpha^{(0)}_{s + t} \ \gamma^\Lambda_{-s-t} = 
\alpha^\Lambda_{s + t} \quad \text{for} \quad s,t \in \RR \, ,
\end{align*}
proving that $\alpha^\Lambda_\RR$ satisfies the group law.

(i) According to Lemma \ref{l3.3}(iii), the 
restrictions of the cocycles $\gamma^\Lambda_\RR$ to
$\cK(\Z)$ act as automorphisms on this subalgebra; moreover 
the function $t \mapsto \gamma^\Lambda_t$ is absolutely
continuous according to the second part of this lemma. 
The restriction of $\alpha^{(0)}_\RR$ to 
$\cK(\Z)$ induces likewise an automorphic action 
on this subalgebra which is pointwise norm continuous.
This follows from the fact that each algebra $\cC(\lambda)$, 
$\lambda \in \Z$, is left invariant by the action 
of $\alpha^{(0)}_\RR$ and the
elements of these algebras are represented by 
compact operators in the (faithful) 
Schr\"odinger representation. 
Combining the two actions, the statement follows.
(Note that the action of $\alpha^\Lambda_\RR$ 
on the full resolvent algebra is discontinuous, however.)

(ii) It was shown in Lemma \ref{l3.3}(iv) that 
$\gamma_t \doteq \lim_{\Lambda \nearrow \Z} \gamma^\Lambda_t$  exists 
pointwise on $\cR(\Z)$ in the norm topology, 
satisfies the cocycle relation \eqref{r3.1}, 
and $t \mapsto \gamma_t$ is pointwise norm continuous, 
$t \in \RR$. Hence the limits  
$\alpha_t \doteq  \lim_{\Lambda \nearrow \Z} \, \alpha^{(0)}_t \, \gamma^\Lambda_t $,
$t \in \RR$, 
exist as well in this topology and thereby define a group of automorphisms 
$\alpha_\RR$ of the C*-algebra $\cR(\Z)$. Since 
$\cK(\Z)$ is a closed subalgebra of $\cR(\Z)$, this strong form of  
convergence also implies that the restriction 
of $\alpha_\RR$ to $\cK(\Z)$ 
defines automorphisms of it. Moreover, the continuity properties 
of $t \mapsto \gamma_t$ 
and $t \mapsto \alpha^{(0)}_t$ imply that their combined action 
$t \mapsto \alpha_t = \alpha^{(0)}_t \gamma_t$ is pointwise norm 
continuous on $\cK(\Z)$ as well,
$t \in \RR$. 
This completes the proof of the proposition. \qed
\end{proof}

\section{Equilibrium states}
\setcounter{equation}{0}

In this section we turn to the analysis of states on the resolvent algebra
with the aim to exhibit equilibrium (KMS) states for the dynamics 
constructed in the preceding section; cf.\ also \cite{KaMa} for a recent 
investigation of KMS states 
on the resolvent algebra in the one-dimensional case.   
In a large part of our study we will restrict attention to states on the 
subalgebra $\cK(\Z)$. Only at the very end of our analysis we will
extend the locally normal states to the full 
algebra. The reason for this approach is twofold: first, we can work in the
setting of C*-dynamical systems, where we can profit from many known 
results; second, we can illustrate the advantages of the present
framework, where we deal with an algebra which has a rich ideal structure,
allowing to describe finite and infinite systems at the same 
time. Thus, instead of changing the algebra when proceeding to 
theories in ``larger boxes'', we change states on our fixed algebra. This 
will lead to considerable simplifications when proceeding to the 
thermodynamic limit.

We begin with a brief account of notation and terminology: States 
$\omega$ are positive, 
linear and normalized functionals on a  
unital C*-algebra. By the GNS construction,
any such state yields a concrete representation (homomorphism) 
$\pi$ of the algebra 
into the bounded operators on some Hilbert space $\cH$, and a unit vector 
$\Omega \in \cH$ such that \ 
$\omega( \, \bdot \,) = \langle \Omega, \pi( \, \bdot \, ) \Omega \, \rangle$,
where the dot stands for any given member of the algebra. 
Restricting attention to states on the algebra $\cK(\Z)$, we 
adopt the following terminology.  

\vspace*{2mm}
\noindent 
\textbf{Definition:} Let $\omega$ be a state on $\cK(\Z)$ and
let $\Lambda \subset \Z$ be any finite subset. 
\begin{itemize}
\item[(i)] 
$\omega$ is \textit{normal} at $\Lambda$ 
if for every approximate identity formed by increasing projections 
$\{ E_\iota \in \cC(\Lambda) \}_{\iota \in \II} $ one has
$\lim_\iota \omega(E_\iota) = 1$. The state is said to be 
\textit{locally normal} at some infinite subset $\bLambda \subseteq \Z$
if it is normal at all finite subsets $\Lambda \subset \bLambda$.
(Bold face symbols~$\bLambda$ will be used for infinite 
subsets of $\Z$.) 
\item[(ii)] 
$\omega$ is said to be singular at $\Lambda$ if
$\omega \upharpoonright \cC(\Lambda) = 0$. (Such
states exist since $\cC(\Lambda)$ forms an ideal 
of compact operators in $\cK(\Lambda)$ and
$\cK(\Z) = \cK(\Lambda) \otimes \cK(\Z \backslash \Lambda)$.) 
\end{itemize}

\medskip
The following technical lemma concerning such states 
will be used at various points in our analysis. 
\begin{lemma} \label{l4.1}
Let $\omega$ be a state on $\cK(\Z)$ with 
GNS-representation $(\pi, \cH, \omega)$. 
\begin{itemize}
\item[(i)] If, for some $\Lambda \subset \Z$, the state $\omega$
is normal at all points $\lambda \in \Lambda$, it is normal at~$\Lambda$.
The restriction of the 
representation $\pi \upharpoonright \cC(\Lambda)$ can then be extended
in the strong operator topology 
to a regular representation of the resolvent algebra $\cR(\Lambda)$ which
is quasi-equivalent to the unique Schr\"odinger representation of this
algebra.
\item[(ii)] If the state $\omega$ is singular at some point 
$\lambda \in \Z$, it is singular at any $\Lambda \ni \lambda$ 
and $\pi \upharpoonright \cC(\Lambda) = 0$ for such $\Lambda$.
\end{itemize}
\end{lemma}
\begin{proof}
(i) Let $\Lap,  \Lapp \subset \Z$, 
let 
$A(\lap) \in \cK(\lap)$, $\lap \in \Lap$,
and let 
$B(\lapp) \in \cK(\lapp)$,
$\lapp \in \Lapp$. 
The tensor product structure of $\cK(\Z)$   
and the normality of $\omega$ at $\lambda \in \Lambda$ 
imply that for any  
approximate identity formed by increasing projections 
{$\{ E_\iota(\lambda) \in \cC(\lambda) \}_{\iota \in \II} $},  one has 
\begin{align*}
& \lim_\iota \, \omega(\Pi_{\lap} A(\lap) \, E_\iota(\lambda) 
\, \Pi_{\lapp} B_\kappa(\lapp)) \\
& = \lim_\iota \, \omega(\Pi_{{\lap \neq \lambda}} 
A(\lap) 
\Pi_{\lapp \neq \lambda} B_\kappa(\lapp)
\, A(\lambda)  E_\iota(\lambda)  B(\lambda) ) \\
& = \omega(\Pi_{{\lap \neq \lambda}} 
A(\lap) 
\Pi_{\lapp \neq \lambda} B_\kappa(\lapp)
\, A(\lambda) B(\lambda) ) 
 = \omega(\Pi_{\lap} A(\lap) \ 1 \ 
\Pi_{\lapp}  B_\kappa(\lapp)) \, .
\end{align*}
The second equality obtains since 
$E_\iota(\lambda) \in \cC(\lambda)$,
$\iota \in \II$, is an approximate identity of increasing projections
and $\lim_\iota \omega(E_\iota(\lambda)) = 1$. It follows from this 
relation that \ $\lim_\iota \pi(E_\iota(\lambda)) = 1$ 
in the weak, hence also strong operator topology. 
 
Picking approximate identities from each 
$\cC(\lambda)$, $\lambda \in \Lambda$,
their product defines an approximate identity of increasing projections  
$\{ E_\iota(\Lambda) \in \cC(\Lambda) \}_{\iota \in \II} $ 
which consequently converges to $1$ in the strong operator topology,
as well. Given any other approximate identity 
$\{ F_\kappa(\Lambda) \in \cC(\Lambda) \}_{\kappa \in \KK} $
one has 
$$
\pi(E_\iota(\Lambda)) = \lim_\kappa 
\pi(F_\kappa(\Lambda) E_\iota(\Lambda) F_\kappa(\Lambda))
\leq \lim_\kappa \pi(F_\kappa(\Lambda)) \, ,
$$ 
hence the latter approximate identity
also converges to $1$ in the strong operator topology.
Thus $\omega$ and $\pi \upharpoonright \cC(\Lambda)$ are normal
at $\Lambda$. 

Now let 
$\{ E_\iota(\Lambda) \in \cC(\Lambda) \}_{\iota \in \II} $ 
be an approximate identity, \ie 
$\lim_\iota \pi(E_\iota(\Lambda)) = 1$ in the strong
operator topology. Clearly, 
$R E_\iota(\Lambda),   E_\iota(\Lambda) R \in  \cC(\Lambda)$, 
$\iota \in \II$, 
for any $R \in \cR(\Lambda)$ since $\cC(\Lambda) \subset  \cR(\Lambda)$
is an ideal. Moreover, the corresponing limits 
$\pi(R) \doteq \lim_\iota \pi(R E_\iota(\Lambda)) = 
\lim_\iota \pi(E_\iota(\Lambda) R)$ exist in the strong operator
topology, defining an extension of $\pi \upharpoonright \cC(\Lambda)$
to $\cR(\Lambda)$, denoted by the same symbol. For the proof
of regularity, let $R(c) \in \cR(\Lambda)$, 
$c \in \RR \backslash \{0\}$, be any of the basic resolvents  
defined in relation \eqref{r2.1}. Then 
$\lim_{c \rightarrow \infty} \, ic R(c) E_\iota(\Lambda) = 
E_\iota(\Lambda)$, $\iota \in \II$, in the norm topology.  
Since the function $c \mapsto \pi(i c R(c))$ is uniformly
bounded, it follows by a three epsilon argument that 
\begin{align*}
\lim_{c \rightarrow \infty} \, \pi(i c R(c)) & =
\lim_{c \rightarrow \infty} \lim_\iota 
\, \pi(i c R(c) \, E_\iota(\Lambda))  \\
& =
\lim_\iota \lim_{c \rightarrow \infty} \, \pi(i c R(c,f) E_\iota(\Lambda))
= 1
\end{align*}
in the strong operator topology. The 
regularity of the representation then
follows, cf.\ \cite[Prop.~4.5]{BuGr2}. The statement 
about quasi-equivalence is a consequence of von Neumann's 
uniqueness theorem for finite regular quantum systems. 

(ii) Let $\omega$ be singular at some given $\lambda \in \Z$ and
let $\Lambda \ni \lambda$. Similarly as in the preceding 
step, one shows that for any approximate identity 
$\{ E_\iota(\lambda) \in \cC(\lambda) \}_{\iota \in \II} $ 
one has $\pi(E_\iota(\lambda)) = 0$, $\iota \in \II$.
Hence if $\Lambda \ni \lambda$, the product 
$\{ E_\iota(\Lambda) \in \cC(\Lambda) \}_{\iota \in \II} $ of
any family of approximate identities 
$\{ E_\iota(\lap) \in \cC(\lap) \}_{\iota \in \II} $, 
$\lap \in \Lambda$, satisfies 
$\pi(E_\iota(\Lambda)) = 0$, $\iota \in \II$. Since 
$\pi(C) = \lim_\iota\pi(C E_\iota(\Lambda)) = 0$ for $C \in \cC(\Lambda)$,
this completes the proof of the lemma. \qed
\end{proof}

We turn next to the study of the dynamics in given states
and their GNS-representations.
In addition to the free dynamics $\alpha^{(0)}_\RR$ and the interacting
dynamics $\alpha_\RR$ we will also consider 
the approximating dynamics $\alpha^\Lambda_\RR$
for $\Lambda \subset \Z$. The following lemma is of some interest
in its own right since it displays the role of singularities 
of states on $\cK(\Z)$ as ``impenetrable boundaries'': 
\ie if a state is singular in some region, any form of
interaction is turned off there in its GNS representation.

\begin{lemma} \label{l4.2}
Let $\Lambda \subset \Z$, let $\omega$ be a state on 
$\cK(\Z)$ which is normal at $\Lambda$ and 
singular at all 
points $\lap \in \Z \backslash \Lambda$,  and
let $(\pi, \cH, \Omega)$ be its GNS-representation. Then
\begin{itemize}
\item[(i)] \ $\pi(\cK(\Z)) = \pi(\cK(\Lambda))$ \ and \ 
$\pi \upharpoonright \cC(\Lap) = 0$ \ if \
$\Lap \cap (\Z \backslash \Lambda) \neq \emptyset $ 
\item[(ii)] \ $\pi(\alpha_t(C^\prime)) = 0$, $t \in \RR$, for any 
$C^\prime \in \cC(\Lap)$ \ if \   
$\Lap \cap (\Z \backslash \Lambda) \neq \emptyset $ 
\item[(iii)] \ $\pi(\alpha_t(C)) = \pi(\alpha_t^\Lambda(C))$, $t \in \RR$, for
any $C \in \cK(\Z)$.  
\end{itemize}
\end{lemma}
\begin{proof}
(i) The first statement is an immediate consequence of the preceding lemma
and the definition of the algebra $\cK(\Z)$, which is generated by 
the algebras $\cC(\Lambda)$, $\Lambda \subset \Z$, and the unit 
operator.  

(ii) The crucial point is the observation that the time evolution 
$\alpha_\RR$ maps the algebra~$\cC(\Lap)$ into a norm closed subalgebra 
$\cS(\Lap) \subset \cK(\Z)$, the surrounding of 
$\cC(\Lap)$, which is generated by all algebras 
$\cC(\Lapp)$ with $\Lapp \supseteq \Lap$. According to  
Lemma \ref{l4.1}(ii) we have $\pi \upharpoonright \cC(\Lapp) = 0$
if $\Lapp \supset \Lap$, whence by continuity also 
$\pi  \upharpoonright \cS(\Lap) = 0$. The statement then follows. 

Now the harmonic dynamics $\alpha^{(0)}_\RR$
leaves all algebras $\cC(\Lapp)$, $\Lapp \subset \Z$,
invariant, so it suffices to determine the action of the 
cocycle $\gamma_\RR$. Given 
$C^\prime \in \cC(\Lap)$, we expand $\gamma_t(C^\prime)$ for small
$|t|$ into a (norm-convergent) Dyson series; it is given by  
$ \gamma_t(C^\prime) = C^\prime + \sum_{n=1}^\infty i^n D_n^{\Lap_n}(t)(C^\prime) $, 
where we used Lemma~\ref{l3.2}(i) by choosing for 
the $n$-th order contribution $ D_n^{\Lap_n}(t)(C^\prime)$ 
some sufficiently big $\Lap_n \supset \Lap$. 
Decomposing $D_n^{\Lap_n}(t)(C^\prime)$  
into the maps $\Ms_n(t)(C^\prime)$, it follows from 
Lemma~\ref{l3.1}(iii) that 
$D_n^{\Lap_n}(t)(C^\prime) \in \bigcup_{\Lambda^\prime \subseteq 
\Lapp \subseteq \Lap_n} \cC(\Lapp)$. Hence, 
one has $\alpha_t(C^\prime) = \alpha^{(0)}_t \gamma_t(C^\prime) 
\in \cS(\Lambda^\prime)$ for small $|t|$ and, making use of the group law
for $\alpha_\RR$, this result extends to all~$t \in \RR$. 

(iii) Let $\Lambda_0 \subset \Z$ and 
let $C \in \cC(\Lambda_0)$. If 
$\Lambda_0 \cap  (\Z \backslash \Lambda) \neq \emptyset$, the
statement follows from the preceding step, so we may assume that 
$\Lambda_0 \subseteq \Lambda$. 
Similarly to the preceding step, we expand
$\gamma_t(C)$ for small $|t|$ into a Dyson series, 
$ \gamma_t(C) = C + \sum_{n=1}^\infty i^n D_n^{\Lambda_n}(t)(C)$, 
where each $\Lambda_n \supset \Lambda_0$ is sufficiently big.
Decomposing $D_n^{\Lambda_n}(t)(C)$ into the terms  
$\Ms_n(t)(C)$, it follows again from Lemma \ref{l3.1}(iii)
that for each term one has $\Ms_n(t)(C) \in \cC(\Lambda_{\, n}^{\tlsml})$,
where $\Lambda_{\, n}^{\tlsml} \supseteq \Lambda_0$. 
If $\Lambda_{\, n}^{\tlsml} \cap  (\Z \backslash \Lambda) \neq \emptyset$,
the corresponding terms $\Ms_n(t)(C)$ 
are annihilated in the representation~$\pi$ according
to Lemma \ref{l4.1}(ii). In other words, in the representation $\pi$ 
there contribute only terms $\Ms_n(t)(C)$ to the Dyson series with  
$\Lambda_{\, n}^{\tlsml} \subseteq \Lambda$, $n \in \NN$. 
The resulting series thus coincides exactly with the Dyson series 
for $\gamma^\Lambda_t$. 

So we conclude that 
$\pi(\gamma_t(C)) = \pi(\gamma^\Lambda_t(C))$ for 
$C \in \cC(\Lambda_0)$, $\Lambda_0 \subset \Z$, and 
small~$|t|$.
Since the algebra $\cK(\Z)$ is generated by 
the algebras $\cC(\Lambda_0)$, $\Lambda_0 \subset \Z$, 
and the unit operator $1$ and since $\pi$ is a homomorphism,  
this relation holds therefore for all $C \in \cK(\Z)$. Using
the group law for the dynamics $\alpha_\RR$ and $\alpha^\Lambda_\RR$,
and the invariance of $\cK(\Z)$ under both dynamics, the
statement follows. \qed
\end{proof}

This lemma shows how one can describe finite systems on the 
resolvent algebra, using partially singular states. 
We will make use of this fact in our construction of KMS states;
cf.\ \cite{BrRo} for basic definitions in this context. Given any 
$\Lambda \subset \Z$ and $\beta > 0$, there exists a (unique)
normal KMS state $\rho_\beta^\Lambda$ on the
algebra $\cK(\Lambda)$ for the dynamics 
$\alpha^\Lambda_\RR \upharpoonright \cK(\Lambda)$; note that
this subalgebra is stable under the action of $\alpha^\Lambda_\RR$. 
The state $\rho_\beta^\Lambda$ is  
described in the Schr\"odinger representation of this finite
system by the density matrix 
$\text{\footnotesize $Z_\Lambda^{-1}$} \, e^{- \beta H_\Lambda}$,
where 
\begin{equation}
H_\Lambda = 
(1/2) \sum_{\lambda \in \Lambda} \, (P_\lambda^2 + \varpi^2 \, Q_\lambda^2) \ + 
\! \sum_{\lsml \in \Lambda \times 
\Lambda} \, 
 V_{\lfrl}(Q_{\lap}
- Q_{\lapp}) \, .
\end{equation}

Exploiting the tensor product structure of $\cK(\Z)$,
we extend $\rho_\beta^\Lambda$ to the full algebra. This extension is 
given by the product state 
$\omega_\beta^\Lambda \doteq \rho_\beta^\Lambda \otimes \sigma^{\Z \backslash \Lambda}$
on $\cK(\Lambda) \otimes \cK(\Z \backslash \Lambda)$,
where $\sigma^{\Z \backslash \Lambda}$ denotes the singular state 
which annihilates all operators in $ \cK(\Z \backslash \Lambda)$, apart from
multiples of~$1$. The latter state exists 
since $\cK(\Z \backslash \Lambda) = \otimes_{\lap \in \Z \backslash \Lambda} \,
\cK(\lap)$ and since each algebra $\cC(\lap)$
is an ideal in $\cK(\lap) = \CC 1 + \cC(\lap)$, 
$\lap \in \Z \backslash \Lambda$.
The properties of the states $\omega_\beta^\Lambda$, given in the 
subsequent lemma, are a simple consequence of the preceding results.  

\begin{lemma}  \label{l4.3} 
Let $\Lambda \subset \Z$, $\beta > 0$, and let 
$\omega_\beta^\Lambda$ be the state on $\cK(\Z)$, defined above, 
with GNS representation 
$(\pi_\beta^\Lambda, \cH_\beta^\Lambda, \Omega_\beta^\Lambda)$.  
\begin{itemize}
\item[(i)] \ 
$\omega_\beta^\Lambda$ is a KMS state for the full dynamics $\alpha_\RR$
and the given $\beta$
\item[(ii)] \ 
$\omega_\beta^\Lambda$ is normal at $\Lambda$ and singular at any 
point $\lap \in \Z \backslash \Lambda$
\item[(iii)] \ $\omega_\beta^\Lambda$ is a primary state, \ie 
$\pi_\beta^\Lambda(\cC(\Z))^{\prime \prime}$ is a factor. 
\item[(iv)] For given $\Lambda \subset \Z$ and $\beta > 0$, the state 
is uniquely fixed by these three properties.
\end{itemize}
\end{lemma}
\begin{proof}
(i) Let $C_1, C_2 \in \cK(\Z)$ and consider the function
{$t \mapsto \omega_\beta^\Lambda(C_1 \alpha_t(C_2))$, $t \in \RR$}.
According to Lemma \ref{l4.1}(ii) this function vanishes 
whenever one of the operators contains a factor
$C \in \cC(\Lambda_0)$ with 
$\Lambda_0 \, \bigcap \, (\Z \backslash \Lambda) \neq \emptyset$; 
thus it satisfies the KMS condition. Because of the
linearity of $\omega_\beta^\Lambda$  we may therefore restrict attention
to the case, where $C_1, C_2 \in \cK(\Lambda)$. 
It then follows from 
Lemma \ref{l4.2}(iii) and the fact that
$\cK(\Lambda)$ is stable under the action of $\alpha^\Lambda_\RR$
that the functions 
$t \mapsto \omega_\beta^\Lambda(C_1 \alpha_t(C_2)) = 
\omega_\beta^\Lambda(C_1 \alpha^\Lambda_t(C_2))$
satisfy the KMS condition, bearing in mind that 
$\omega_\beta^\Lambda \upharpoonright \cK(\Lambda) =  \rho_\beta^\Lambda$
is by definition a KMS state for the dynamics 
$\alpha^\Lambda_\RR \upharpoonright \cK(\Lambda)$.
 
(ii) $\omega_\beta^\Lambda \upharpoonright \cK(\Lambda) = \rho_\beta^\Lambda$ is 
represented by a density matrix in the Schr\"odinger representation,
hence normal, and its extension to $\cK(\Z)$ is  singular at all
points of $\Z \backslash \Lambda$ by its very construction.

(iii) According to Lemma \ref{l4.2}(i) one has
$\pi_\beta^\Lambda(\cK(\Z))^{\prime \prime} =
\pi_\beta^\Lambda(\cK(\Lambda))^{\prime \prime}$, and the latter 
algebra is a factor according to von Neumann's uniqueness
theorem.

(iv) The last statement is an immediate consequence of the uniqueness of
the KMS state 
$\omega_\beta^\Lambda \upharpoonright \cK(\Lambda)$ 
for the given dynamics $\alpha^\Lambda_\RR \upharpoonright \cK(\Lambda)$. 
\qed \end{proof}

According to this lemma, there exists for every finite subset 
$\Lambda \subset \Z$ and inverse temperature $\beta > 0$ 
some primary KMS state $\omega_\beta^\Lambda$ 
for the fixed C*-dynamical system 
$(\cK(\Z), \alpha_\RR)$ which is normal at
$\Lambda$ and singular at all points of $\Z \backslash \Lambda$.
In order to gain control on the normality properties of 
these states in the thermodynamic limit, 
we consider KMS states for the perturbed 
dynamics, where the interaction between
any given point $\lambda \in \Lambda$ and its complement 
$\Lambda \backslash \lambda$ is turned off. 
To this end we proceed from the KMS-state
$\rho_\beta^\Lml$ on the algebra $\cK(\Lml)$
and extend it to a state on 
$\cK(\Lambda) = \cK(\Lml) \otimes \cK(\lambda)$,
putting $\rho_\beta^{\Lml \, , \, \lambda} \doteq 
\rho_\beta^\Lml \otimes \rho_\beta^\lambda$.
In a regular representation of $\cK(\Lambda)$, 
where $\rho_\beta^\Lambda$ is described by the 
density matrix $\text{\footnotesize $Z_\Lambda^{-1}$} \, e^{- \beta H_\Lambda}$,
the state $\rho_\beta^{\Lml \, , \, \lambda}$
is described  by the density matrix 
$\text{\footnotesize $Z_{\Lml \, , \, \lambda}^{-1}$} \, 
e^{- \beta H_{\Lml \, , \, \lambda}}$, where
$$
H_{\Lml \, , \, \lambda} = H_\Lambda - \bVl
 \quad \text{with} \quad 
\bVl \doteq 
\sum_{\lap \in \Lambda} (V_{\lambda \frown \lap} + V_{\lap \frown \lambda})
\, .
$$ 
Since 
both states are represented by density matrices in the same normal 
representation of $\cK(\Lambda)$, we can proceed to the
weak closure of this algebra. 
Thus, we can make use of the perturbation theory for KMS states,  
developed by Araki \cite{Ar}, where we rely on the 
presentation given in \cite{DeJaPi}. Applying these results, we obtain
the following estimate for the norm difference between
the two states.

\begin{lemma} \label{l4.4} 
Let $\Lambda \subset \Z$, $\lambda \in \Lambda$, and $\beta > 0$.
Moreover, let  $\rho^\Lambda_\beta$, $\rho^{\Lml \, , \, \lambda} _\beta$
be the two states on~$\cK(\Lambda)$, defined above. Then
\begin{itemize}
\item[(i)] \ $\| \rho^\Lambda_\beta - \rho^{\Lml \, , \, \lambda} _\beta
\| \leq 2 \, e^{ \beta \, \| \sbVl \|/2 } \, 
(e^{ \beta \, \| \sbVl  \| /2} - 1) $. \
(Note that this bound does not depend on the size of $\Lambda$.)
\item[(ii)] \ The norm difference of the states
does not change if one extends them 
to the algebra $\, \cK(\Z) = \cK(\Lambda) \otimes 
\cK(\Z \backslash \Lambda)$, putting \ 
$\omega^\Lambda_\beta \doteq \rho^\Lambda_\beta \otimes \sigma^{\Z \backslash \Lambda}$
and $\omega^{\Lml \, , \, \lambda}_\beta 
\doteq \rho^{\Lml \, , \, \lambda}_\beta \otimes \sigma^{\Z \backslash \Lambda}$,
respectively. Here 
$\sigma^{\Z \backslash \Lambda}$ denotes the singular state which 
satisfies $\sigma^{\Z \backslash \Lambda} \upharpoonright \cC(\lap) = 0$
for $\lap \in \Z \backslash \Lambda$. 
\end{itemize}
\end{lemma}
\begin{proof}
(i) Let $(\pi, \cH, \Omega)$
be the GNS representation of $\cK(\Lambda)$, induced by 
$\rho^\Lambda_\beta$. 
According to \cite[Thm.~5.1]{DeJaPi}, there exists then a non-zero vector 
$\Omega_{\sbVl} \in \cH$ such that 
$$
\rho^{\Lml \, , \, \lambda} _\beta (\, \bdot \,) =
\langle \Omega_{\sbVl} , \, \bdot \
\Omega_{\sbVl} \rangle / \| \Omega_{\sbVl} \|^2 \, .
$$
For its norm one has 
$\| \Omega_{\sbVl}  \| \leq \| e^{- \beta  \sbVl/ 2} \, \Omega \| 
\leq e^{\beta \|  \sbVl \| / 2}$ (Golden-Thomson inequality). 
Moreover, one obtains the bound 
$ \| \Omega_{\sbVl} - \Omega \| \leq (e^{ \beta \|  \sbVl \| / 2} - 1)$,
cf. \ \cite[Thm.~5.3]{DeJaPi}. These estimates are not optimal, but
they suffice for our purposes. For they imply
\begin{align*}
\| \, \rho^\Lambda_\beta - \rho^{\Lml \, , \, \lambda} _\beta \, \| 
& = \| \langle \Omega, \pi(\, \bdot \ )  \, 
\Omega \rangle -  \langle \Omega_{\sbVl}, \pi (\, \bdot \ ) \Omega_{\sbVl} 
\rangle / \| \Omega_{\sbVl} \|^{-2} \, \| \\
& \leq  \| \Omega -  \Omega_{\sbVl} \| \, \big(1 + \| \Omega_{\sbVl} \| \big)
+ \big| 1 - \|\Omega_{\sbVl} \|^2 \big| \\
& \leq 2 \,  \| \Omega -  \Omega_{\sbVl} \| 
\, \big(1 + \| \Omega_{\sbVl} \| \big) \, ,
\end{align*}
from which the first statement follows. 

(ii) The second statement is an immediate consequence of Lemma \ref{l4.2}(i),
completing the proof of this lemma.
\qed \end{proof}

Let $\beta > 0$ be fixed and let 
$K_\beta$ be the weak-*-closure of the convex hull of the corresponding 
KMS states on $\cK(\Z)$ for the dynamics $\alpha_\RR$. 
Recalling that $(\cK(\Z), \alpha_\RR)$  
is a C*-dynamical system, it follows from 
\cite[Prop.\ 5.3.30]{BrRo} that all extremal points of $K_\beta$ are 
primary (factorial) KMS states, which are
disjoint. In our proof that KMS states exist in $K_\beta$ which are 
normal at all points of $\Z$ we will rely on the following basic fact. 

\begin{lemma} \label{l4.5}
Let $\omega$ be a primary state on $\cK(\Z)$ with factorial 
GNS representation $(\pi, \cH, \Omega)$. There
is a (empty, finite, or infinite) set 
$\bLambda \subseteq \Z$  such that 
$\omega$ is normal at all points $\lambda \in \bLambda$ and
singular at all points $\lap \in \Z \backslash \bLambda$.
\end{lemma}
\begin{proof}
Let $\lambda \in \Z$, and let the projections 
$\{ E_\iota(\lambda) \in \cC(\lambda) \}_{\iota \in \II}$ form an 
approximate identity of $\cC(\lambda)$. 
Making use of the arguments in the proof of 
Lemma \ref{l4.1}(i),  one shows
that $\lim_\iota  \pi(E_\iota(\lambda))$ lies in the
center of $\pi(\cK(\Z))^{\prime \prime}$. Since $\pi$ is factorial
it follows that $\lim_\iota  \pi(E_\iota(\lambda)) \in \{0,1 \}$.
Let $\bLambda = \{ \lambda \in \Z : 
\lim_\iota  \pi(E_\iota(\lambda)) = 1 \}$. Then 
$\lim_\iota  \pi(E_\iota(\lap)) = 0$ for $\lap \in 
\Z \backslash \bLambda$, completing the proof. 
\qed \end{proof}

We are equipped now with the tools for the  proof that there exist 
primary states in $K_\beta$ which are locally normal at all points of $\Z$.
At this point we have to assume that $\beta > 0$ is sufficiently small,
\ie the temperature is sufficiently high. For we will need that the upper 
bound on the norm difference of states, given in Lemma \ref{l4.4},
is less than $1$. This is accomplished if \ 
$\beta \, \| \bVl \| < 2 \, | \ln(\sqrt{3} - 1) |$ for $\lambda \in \Z$. 
Note that for given potential $V$
and dimension $d$ of the lattice one has~$\| \bVl \| \leq 2^d \, \| V \|$.

\begin{lemma} \label{l4.6} 
Let $V$, $\beta$ be given, where  
$0 < \beta \, \| V \| < 2^{1-d} \,  | \ln(\sqrt{3} - 1) | $, 
and let $\Lambda_m$, $m \in \NN$, be 
an increasing family of subsets, covering $\Z$.
Moreover, let $K_{\beta, n} \subset K_\beta$ be the weak-$^*$-closures of
the convex hull of states on  $\cK(\Z)$ given by 
$
\{ \sum_{m \geq n} p_m \, \omega_\beta^{\Lambda_m} :
p_m \geq 0, \ \sum_{m = n}^\infty p_m = 1 \} \, , \ 
n \in \NN \, 
$.  
Their intersection $\bigcap_{\, n \in \NN} K_{\beta, n}$ is not empty and its
extremal points are primary KMS states which are locally normal on~$\Z$.
\end{lemma}
\begin{proof}
We begin by noting that $\bigcap_{n \in \NN} K_{\beta, n}$ is not empty
since it is closed in the weak-*-topology and $1 \in \cK(\Z)$.
Moreover, by central decomposition \cite[Prop.~5.3.30]{BrRo} one can 
decompose the states in this set into primary, hence extremal 
KMS states for the given
$\beta$. Next, let $\lambda \in \Z$ and let $n \in \NN$ be sufficiently large 
such that $\lambda \in \Lambda_m$, $m \geq n$. Then, for such~$m$ 
and corresponding primary states $\omega_\beta^{\Lambda_m}$,  one has 
according to Lemma \ref{l4.4}(ii) and the choice of~$\beta \, \| V \|$,
bearing also in mind that by construction 
$\omega_\beta^{\Lambda_m \backslash \lambda \, , \, \lambda} \upharpoonright 
\cK(\lambda) = \omega_\beta^\lambda  \upharpoonright 
\cK(\lambda)$,  
$$
\| (\omega_\beta^{\Lambda_m} - \omega_\beta^\lambda) \upharpoonright 
\cK(\lambda) \| 
\leq \| \omega_\beta^{\Lambda_m} - 
\omega_\beta^{\Lambda_m \backslash \lambda \, , \, \lambda}  \| \leq c < 1 \, . 
$$ 
It follows that for any convex combination of the states 
$\omega_\beta^{\Lambda_m}$, $m \geq n$, one has 
$$
\| \, 
(\sum_{m \geq n} p_m \, \omega_\beta^{\Lambda_m} - \omega_\beta^\lambda) \upharpoonright
\cC(\lambda) \| 
\leq \sum_{m \geq n} p_m \, \|(\omega_\beta^{\Lambda_m} - \omega_\beta^\lambda) 
\upharpoonright \cC(\lambda) \|  \leq c < 1 \, ,
$$
and this estimate holds for their weak-*-limit points in $K_{\beta,n}$,
as well. 

Now let $\omega_\beta \in \bigcap_{n \in \NN} K_{\beta, n}$ be any primary KMS
state. It follows from this estimate that 
$\| \omega_\beta \upharpoonright \cC(\lambda) \| 
\geq \|\omega_\beta^\lambda  \upharpoonright \cC(\lambda) \| 
- \| (\omega_\beta - \omega_\beta^\lambda)  \upharpoonright \cC(\lambda) \| 
> 0$, 
$\lambda \in \Z$. Hence $\omega_\beta \upharpoonright \cC(\lambda)$ does 
not vanish at any point $\lambda \in \Z$ and therefore is locally normal
according to Lemma \ref{l4.5}.
\qed \end{proof}

At this point we can finally 
proceed from the algebra $\cK(\Z)$ back to the full
resolvent algebra~$\cR(\Z)$, making use of Lemma \ref{l4.1}(i). We
thereby arrive at the main result of this section.

\begin{proposition} \label{p4.7}
Let $V \in \Co(\RR^d)$ be any interaction potential between 
nearest neighbors on $\Z$ and 
let $0 < \beta \, \| V \| < 2^{1 - d} \, |\ln \, (\sqrt{3} - 1 ) \, |$. 
There exist regular and primary KMS states 
$\omega_\beta$ on the resolvent algebra $\cR(\Z)$ for the 
corresponding dynamics $\alpha_\RR$. 
\end{proposition}
\begin{proof}
According to Lemma \ref{l4.6} there exists some locally normal 
and primary KMS state $\omega_\beta$ on $\cK(\Z)$ with 
GNS representation $(\pi_\beta, \cH_\beta, \Omega_\beta)$. Since the 
state satisfies the KMS condition, there exists on $\cH_\beta$ 
a continuous unitary representation $U_\beta(t)$ of the time translations
such that $U_\beta(t) \, \pi_\beta(\, \bdot \,) =
\pi_\beta \circ \alpha_t ( \, \bdot \, ) \, U_\beta(t)$, \ie 
the representations $\pi_\beta$ and $\pi_\beta \circ \alpha_t$
are unitarily equivalent,  $t \in \RR$. Moreover, the state 
$\omega_\beta$ and hence $\pi_\beta$ is locally
normal, so the same is true for the representations 
$\pi_\beta \circ \alpha_t$, $t \in \RR$. We can therefore 
uniquely extend these representations to regular representations of 
the  C*-inductive limit $\cR(\Z)$, 
formed by the algebras $\cR(\Lambda)$, $\Lambda \subset \Z$, 
cf. Lemma \ref{l4.1}(i). 
Since $\cR(\Z)$ is stable under the action of $\alpha_\RR$,
these extensions are unitarily equivalent with the same 
unitary intertwiners $U_\beta(t)$, $t \in \RR$. 
The statement then follows. 
\qed \end{proof}

\section{Ground states}
\setcounter{equation}{0}

Having established the existence of regular KMS states for large 
temperatures on the resolvent algebra, let us turn now to the opposite 
end of the temperature scale, \ie the ground states. The construction of
such states proceeds in exactly the same manner as that of KMS states 
with one major difference: The perturbation theory 
underlying Lemma \ref{l4.4} is not applicable in this 
case, so one needs a different strategy in order to establish regularity. 
We will rely here on arguments, given in \cite[Lem.~7.3]{BuGr2} for 
one-dimensional lattice systems, which work in the 
present case, as well.  

We proceed again from the C*-dynamical system $(\cK(\Z), \alpha_\RR)$
and consider regular ground 
states $\rho_\infty^\Lambda$ on $\cK(\Lambda)$
for the approximating dynamics 
$\alpha_\RR^\Lambda \upharpoonright \cK(\Lambda)$, $\Lambda \subset \Z$,
corresponding to 
zero temperature $\beta = \infty$. These states can be represented
in the Schr\"odinger representation by normalized eigenvectors $\Omega_\Lambda$
for the generators of the dynamics, 
\begin{equation}
H_\Lambda = (1/2) \sum_{\lambda \in \Lambda} \, (P_\lambda^2 + \varpi^2 \, 
Q_\lambda^2) + \! 
\sum_{\lsml \in \Lambda \times \Lambda} \! V(Q_\lap - Q_\lapp) \, - \, E_\Lambda \, 1 \, ,
\end{equation} 
where the constants $E_\Lambda$ are chosen in such a way that 
$H_\Lambda \Omega_\Lambda = 0$, $\Lambda \in \Z$; note 
that $H_\Lambda$ is positive and has discrete spectrum. 
The states $\rho_\infty^\Lambda$ are extended to the full algebra
$\cK(\Z) = \cK(\Lambda) \otimes \cK(\Z \backslash \Lambda)$, 
putting $\omega_\infty^\Lambda \doteq \rho_\infty^\Lambda
\otimes \sigma^{\Z \backslash \Lambda}$, where 
$\sigma^{\Z \backslash \Lambda}$ is the singular state, invented in the 
preceding section. 

Now Lemmas \ref{l4.1}, \ref{l4.2} and \ref{l4.5} are applicable without 
changes in the present context; in Lemma~\ref{l4.3} one has to 
put $\beta = \infty$ and, in its part (iii), to replace the term 
``primary state'' by the more specific term ``pure state''. 
Thus the states $\omega_\infty^\Lambda$, $\Lambda \subset \Z$,
are all pure ground states for the C*-dynamical system 
$(\cK(\Z), \alpha_\RR)$. In analogy to the preceding discussion,
we consider the weak-*-closed convex (hence compact) set $K_\infty$
formed by these ground states and note that all of its extremal points
are pure ground states \mbox{\cite[Prop.~5.3.37]{BrRo}}. 

For the proof that $K_\infty$ contains locally normal states, 
we make use of the fact that the operators 
$(\mu 1 + H_\lambda)^{-1}$, $\mu > 0$, are compact in the 
Schr\"odinger representation of $\cC(\lambda)$ 
and hence are elements of these algebras, $\lambda \in \Z$.
The following statement replaces Lemma \ref{l4.4} in case 
of ground states.

\begin{lemma}
Let $\Lambda \subset \Z$ and $\lambda \in \Lambda$. Then
$$
\omega_\infty^\Lambda((\mu 1 + H_\lambda)^{-1}) \geq   
(\mu + 2^{d+2} \|V \|)^{-1} \, , \quad \text{where} \quad
\mu > 0 \, .
$$
(Note that the lower bound does not depend on the size of $\Lambda$.)
\end{lemma}
\begin{proof}
Proceeding to the Schr\"odinger representation induced 
by the state $\rho_\infty^\Lambda$ on $\cK(\Lambda)$, we 
decompose the generator of the underlying dynamics, 
$$
H_\Lambda = H_{\Lml} + H_\lambda + 
\sum_{\lap \in \Lambda} (V_{\lap \frown \lambda} +
V_{\lambda \frown \lap}) + 
(E_{\Lml} + E_{\lambda} - E_\Lambda) \, 1 \, .
$$ 
Since $H_{\Lml}$ and $H_\lambda$ commute, they have a 
common normalized ground state vector $\Omega$, \ie 
$H_{\Lml} \, \Omega = H_\lambda \, \Omega = 0$. As 
$\| \sum_{\lap \in \Lambda} (V_{\lap \frown \lambda} +
V_{\lambda \frown \lap}) \| \leq 2^{d + 1} \, \| V \|$ and
$H_\Lambda$ is a positive operator which has $\Omega$ in its domain,
this gives 
$$
0 \leq \langle \Omega, H_\Lambda \Omega \rangle
\leq  2^{d + 1} \, \| V \| + (E_{\Lml} + E_{\lambda} - E_\Lambda) \, .
$$  
Bearing in mind that $H_{\Lml} \geq 0$, it follows 
that $H_\Lambda \geq (H_\lambda - 2^{d+2} \, \| V \|)$ and hence 
$$ 
(\nu 1 + H_\Lambda )^{-1} \leq 
((\nu - 2^{d+2} \, \| V \|) 1 + H_\lambda)^{-1}
$$ 
for $\nu > 2^{d+2} \, \| V \|$. Taking the expectation value of 
these operators in the ground state for~$H_\Lambda$ and putting 
$\nu = (\mu + 2^{d+2}) \, \| V \|$, one finally obtains 
the lower bound 
$\rho_\infty^\Lambda((\mu 1 + H_\lambda)^{-1}) 
\geq  (\mu + 2^{d+2} \|V \|)^{-1}$. 
But the chosen extension  $\omega_\infty^\Lambda$ of $\rho_\infty^\Lambda$ 
to the algebra $\cK(\Z)$ does not affect this lower bound,
so the proof of the lemma is complete.
\qed \end{proof}

We can proceed now exactly as in the construction of 
the thermodynamic limits of equilibrium states.
Let $\Lambda_n \subset \Z$,  $n \in \NN$, be increasing sets,
covering all of $\Z$ in the limit and let 
$K_{\infty , n} \subset K_\infty $ be the weak-*-closures of the convex
hulls of ground states $\omega_\infty^{\Lambda_m}$ with $m \geq n$. 
In complete analogy to Lemma \ref{l4.6}, we have the following result.

\begin{lemma} \label{l5.2}
The intersection $\bigcap_{n \in \NN} \, K_{\infty , n} $ is not empty 
and the extremal points of this set 
are pure ground states which are locally normal on $\Z$. 
\end{lemma}
\begin{proof}
The intersection $\bigcap_{n \in \NN}  K_{\infty , n} $ is not 
empty because of the reasons  
given in the proof of Lemma \ref{l4.6}, so let 
$\omega_\infty \in  \bigcap_{n \in \NN} \, K_{\infty , n}$ be an extremal 
(hence pure) state in this closed, convex set. Given any 
$\lambda \in \Z$ we choose $n \in \NN$ sufficiently large such 
that $\Lambda_m \ni \lambda$, $m \geq n$. It then follows from
the preceding lemma (proceeding to a suitable weak-*-limit of the states 
for $\Lambda \nearrow \Z$) that for $\mu > 0$
one has $\omega_\infty((\mu 1 + H_\lambda)^{-1}) \geq (\mu + 2^{d+2} \| 
V \|)^{-1} > 0$. Since $\omega_\infty$ is pure and thus 
\textit{a fortiori} a primary state on $\cK(\Z)$ and since $\lambda \in \Z$
was arbitrary, it follows from Lemma~\ref{l4.5} that 
$\omega_\infty$ is locally normal on $\Z$. 
\qed \end{proof}

Equipped with this information, one can now establish the following 
proposition. Since the proof is identical to the one given for equilibrium 
states, we can omit it.

\begin{proposition} \label{p5.3}
Let $V \in \Co(\RR^d)$ be the interaction potential 
between nearest neighbors on 
the lattice $\Z$. There exist regular and pure ground states 
$\omega_\infty$ for the corresponding  
dynamics $\alpha_\RR$ on the resolvent algebra $\cR(\Z)$. 
\end{proposition}

\section{More dynamics}
\setcounter{equation}{0}

In the preceding analysis we have considered the simplest non-trivial 
examples of infinite lattice systems with nearest neighbor 
interactions. The forces keeping
the particles in a neighborhood of their respective 
lattice points were of harmonic 
nature and their interaction was described by a 
regular class of potentials. We have restricted attention 
to this case in order not to obscure the novel features of 
our construction of dynamics and states. 

In order to indicate how one can deal with the functional
analytic problems appearing in more complicated 
situations, such as singular interactions and non-harmonic
binding forces, 
and to reveal also certain limitations of the present
framework, we outline here three characteristic examples.
Moreover, we consider only the simple  
case of a single particle with one degree of freedom and
deal with its resolvent algebra~$\cR$ in the (faithful)
Schr\"odinger representation. 

The 
canonical momentum and position operators 
of the particle are denoted by $P,Q$ with
their standard domain of essential 
selfadjointness {$\cD \subset L^2(\RR)$}, consisting of
Schwartz test functions on which~$Q$ acts as a multiplication operator. 
The resolvent algebra $\cR$ is then concretely given as 
the C*-algebra which is generated by the resolvents 
\mbox{$R(a,b,c) \doteq (i c 1 + a P + b \, Q)^{-1}$}, 
$a,b \in \RR$, $c \in \RR \backslash \{0\}$; 
it contains the ideal of compact operators $\cC \subset \cR$.
The free Hamiltonian $H_0 = P^{\, 2}$ acts on~$\cD$ and,  
for the sake of simplicity, we consider here only unbounded 
potentials $V = V(Q)$ which are defined on $\cD$ as well
and for which the resulting Hamiltonian $H = (H_0 + V)$ is essentially 
selfadjoint on $\cD$. 

In order to prove that for $R \in \cR$ also all operators 
$e^{itH} R e^{-itH} \in \cR$, \mbox{$t \in \RR$}, we make use of the fact 
that $e^{itH_0} R e^{-itH_0} \in \cR$, $t \in \RR$, since the free 
dynamics maps the set of resolvents
of linear combinations of $P,Q$ onto itself. Thus, if one can show that 
for $R \in \cR$ all differences 
$(e^{itH} R e^{-itH} - e^{itH_0} R e^{-itH_0})$, $t \in \RR$, are compact 
operators, \ie are elements of $\cC \subset \cR$, the assertion follows. 
In \cite{BuGr2} and also in the preceding analysis, this task was 
accomplished by proving that the operators $(e^{itH} e^{-itH_0} - 1)$, 
$t \in \RR$, are compact. This strategy can work if the essential spectra
of $H_0$ and $H$ coincide, but it is bound to fail otherwise. In the 
following we will deal with either case.

\subsection{Singular potentials}

Before discussing some example of a singular potential, we establish a general
result about perturbations of dynamics, thereby avoiding 
the usage of the Dyson series which entered in the preceding analysis. 

\begin{lemma} \label{l6.1} 
Let $V$ be a potential, defined on the domain~$\cD$, which satisfies  
$\int_0^t \! ds \, \| V \, e^{-is H_0} \Phi \|^2 \leq \| C_t \, \Phi \|^2$ 
for $\Phi \in \cD$, 
where $C_t$ are compact operators, $t \geq 0$. Then 
$H \doteq (H_0 + V)$ is essentially selfadjoint on $\cD$ and 
$\big( e^{itH} \, e^{-itH_0} - 1 \big) \in \cC$, $t \in \RR$.
\end{lemma}
\begin{proof}
We omit the proof of essential selfadjointness of $H$, which is 
based on the fact that~$V$ is relatively compact with respect to $H_0$
under the given assumption. The second statement follows from 
the simple estimate for $\Phi, \Psi \in \cD$, 
\begin{align*}
| \langle \Psi, (e^{itH} e^{-itH_0} - 1) \Phi \rangle |^2 
& = | \! \int_0^t \! ds \, \langle  \Psi, e^{is H} V e^{-is H_0} 
\Phi \rangle |^2 \\ 
& \leq   \int_0^t \! ds^\prime \, \| e^{-is^\prime H} \Psi \|^2 \,
\int_0^t \! ds \, \| V e^{-is H_0}  \Phi \|^2 \, 
\end{align*}
which implies \ $ \| (e^{itH} e^{-itH_0} - 1) \Phi \|^2 \leq t \, 
\| C_t \Phi \|^2$, $\Phi \in \cD$. 
Since $C_t$ is compact, this shows that 
$(e^{itH} e^{-itH_0} - 1)$ maps weakly convergent 
null sequences of vectors to strongly convergent null sequences; it
thus is a compact operator, \ie an element of $\cC$, $t \geq 0$. 
Moreover, $ (e^{-itH} e^{itH_0} - 1) = 
- e^{-itH} \, (e^{itH} e^{-itH_0} - 1) \, e^{itH_0}$, 
so this holds true for all $t \in \RR$. 
\qed \end{proof}

Equipped with this technical result we can exhibit now singular potentials
which lead to dynamics of the resolvent algebra $\cR$. We do not consider 
here the most general case, but illustrate the method by a 
significant example. 

\begin{lemma} \label{l6.2}
Let  $0 < \kappa < 1/2$, $g \in \RR$, and let 
$x \mapsto V(x) \doteq g \, |x|^{-\kappa}$ on $\RR \backslash \{0 \}$,
$V(0) = 0$. The corresponding Hamiltonian $H = (H_0 + V)$
is essentially selfadjoint on $\cD$, and one has 
$(e^{itH} e^{-itH_0} - 1) \in \cC$, $t \in \RR$. Hence  $\cR$  is stable 
under the adjoint action of $e^{itH}$, $t \in \RR$. 
\end{lemma}
\noindent \textbf{Remark:}  The admissible singularities of 
$V$ depend on the number of space dimensions. 
\begin{proof} It is apparent that $V$ is defined on $\cD$. 
We will show that the positive quadratic form 
$\int_0^t \! ds \, e^{isH_0} V^2 e^{-isH_0}$ on $\cD \times \cD$ extends to
a compact operator. The statement then follows from the preceding lemma.
Let $\chi$ be a smooth function which is equal to 
$1$ for $|x| \leq 1$ and~$0$ for $|x| \geq 2$. We 
consider approximations of the squared potential, 
$n \in \NN$,  
$$
x \mapsto V^2_{\ n}(x) \doteq g^2 \, |x|^{-2 \kappa} \big( \chi(x/n)
+ (1 - 2 \kappa)^{-1} \, (x/n) \chi^\prime(x/n) \big) \, ,
$$
where the prime $^\prime$ indicates the derivative. The functions
$x \mapsto V^2_{\ n}(x)$ have compact support and are absolutely 
integrable, hence
their Fourier transforms \mbox{$p \mapsto \widehat{V^2_{\ n}}(p)$} 
are entire analytic. Moreover, since 
$$
V^2_{\ n}(x) = 
\mbox{\large $\frac{d}{dx}$} 
\big( (1 - 2 \kappa)^{-1} g^2 |x|^{-2 \kappa} x 
\chi(x/n) \big) \, , \quad x \neq 0 \, , 
$$
the Fourier transforms have a zero at $p=0$. Since 
$\int \! dx \, |x|^{-2 \kappa} \, e^{ipx} = c \, |p|^{2 \kappa - 1}$
for $p \neq 0$ and each $\widehat{V^2_{\ n}}$ is obtained from
$\widehat{V^2}$ by convolution with a test function, 
it is also clear that one obtains the bound 
$|\widehat{V^2_{\ n}}(p)| \leq c_n \, |p|^{2 \kappa - 1}$
for large~$|p|$. After these preparations we can 
analyze now the kernels of the quadratic forms 
$ \int_0^t \! ds \, e^{is H_0} {V^2_{\ n}} e^{-is H_0}$ in 
momentum space, which are defined in the sense of distributions.
Making use of Dirac's bra-ket notation, they are given by
$$
\langle p | \int_0^t \! ds \, e^{is H_0} {V^2_{\ n}}  
e^{-is H_0} \ | q \rangle = i \, \widehat{V^2_{\ n}}(p-q) \,    
(1 - e^{it(p^2 - q^2)})/(p^2 - q^2) \, .
$$
These kernels are square integrable in $p,q$, as one sees by substituting   
{$k \doteq (p-q)$}, $l \doteq (p+q)$,
\begin{align*}
& \int \! dp \! \! \int \! dq \, |\widehat{V^2_{\ n}}(p-q) \,    
(1 - e^{it(p^2 - q^2)})/(p^2 - q^2)|^2 \\ 
& = 8 \int \! dk \! \! \int \! dl \, |\widehat{V^2_{\ n}}(k)|^2
\, \sin^4(kl/2) / k^2 l^2 
= 8 \int \! dk  |\widehat{V^2_{\ n}}(k)|^2 / |k| \! \! 
\int \! dl \, \sin^4(l/2)/l^2 \, .
\end{align*}
The latter integral with respect to $l$ exists and the 
integral with respect to $k$ exists as well since the 
functions $k \mapsto  |\widehat{V^2_{\ n}}(k)|^2 / |k|$ are continuous
due to the zero of $\widehat{V^2_{\ n}}(k)$ at $k = 0$ and the
asymptotic bound $|\widehat{V^2_{\ n}}(k)|^2 / |k| \leq c_n^2 \, |k|^{4\kappa -3}$.
So we conclude that the forms 
$ \int_0^t \! ds \, e^{is H_0} {V^2_{\ n}} e^{-is H_0}$ extend to operators 
in the Hilbert Schmidt class and hence are compact.
Finally, we make use of the fact that the operators 
$(V^2 - {V^2_{\ n}})$ are bounded by construction with norm
satisfying $\| ( V^2  - {V^2_{\ n}} ) \| \leq c^\prime \, n^{- 2 \kappa}$, where 
$c^\prime$ does not depend on $n \in \NN$. It follows that
$ \| \int_0^t \! ds \, e^{isH_0} ( V^2 - V^2_n )  e^{-isH_0} \| 
\leq t c^\prime \, n^{- 2 \kappa} $. 
Hence the form $\int_0^t \! ds e^{isH_0} V^2 e^{-isH_0}$ can be approximated
in norm by compact operators and thus extends to a compact operator
as well. The statement then follows from the preceding lemma.
\qed \end{proof}

\subsection{Unbounded potentials}

We turn now to the discussion of some non-harmon\-ic 
binding force involving an unbounded potential such that 
$H$ has discrete spectrum. As already mentioned, the method
of proof of stability of $\cR$ under the corresponding 
dynamics, used in the preceding subsection, is then bound
to fail. 

\begin{lemma} \label{l6.3}
Let  $0 < \kappa < 1$, $g, x_0 > 0$, and let 
$x \mapsto V(x) \doteq g \, (x^2 + x_0^2)^{\kappa/2}$.
The corresponding Hamiltonian $H = (H_0 + V)$
is essentially selfadjoint on $\cD$, has a compact resolvent, 
and the resolvent 
algebra~$\cR$  is stable under the adjoint action of $e^{itH}$, $t \in \RR$. 
\end{lemma}
\begin{proof}
Again, we do not deal here with the domain and spectral properties of 
$H$ since they are well known, cf.\ \cite[Thm.~XIII.67]{ReSi}.
For the proof of the main part of the statement 
we note that it suffices to establish it for the generating resolvents 
$\biR \doteq R(a,b,c) \in \cR$, 
where $a,b \in \RR$, $c \in \RR \backslash \{ 0 \}$. 
Putting $\Gamma(t) \doteq e^{itH} e^{-itH_0}$, $t \in \RR$,
we have in the sense of bilinear forms on $\cD \times \cD$,
making use of the fundamental theorem of calculus,  
\begin{align*}
 (\Gamma(t) \biR \, \Gamma(t)^{-1} - \biR) 
 = i \int_0^t \! ds \, e^{is H} \, [V, \biR(-s)] \, e^{-is H} \, ,
\end{align*}
where $\biR(-s) \doteq e^{-isH_0} \biR e^{isH_0} = R(a-2s,b,c)$. 
Since the inverse of $\biR$ is linear in the operators
$P, Q$ and $1$, we obtain   
$[V, \biR(-s)] = -i(a - 2s) \, \biR(-s) V^\prime \biR(-s)$,
where 
$x \mapsto V^{\, \prime}(x) = g \kappa \, x \, (x^2 + x_0^2)^{(\kappa/2 -1)}$. 
So we arrive at 
$$
(\Gamma(t) \biR \, \Gamma(t)^{-1} - \biR) =  
\int_0^t \! ds \, (a - 2s) \,  e^{is H}  \biR(-s) V^{\, \prime} \biR(-s)  
\, e^{-is H} \, .
$$
But $x \mapsto V^\prime(x) \in \Co(\RR)$, so it follows
from arguments similar to those given in the proof 
of Lemma~\ref{l2.1} that the operator function 
$$
s \mapsto (a - 2s) \, 
 e^{is H} \, \biR(-s) V^\prime \biR(-s)\, e^{-is H} 
$$ 
is, for almost all $s \in \RR$, a compact 
operator. Since it is also bounded on compact sets of $\RR$,
its integral (defined in the strong operator topology) is a compact 
operator, as well. This shows that for all resolvents $\biR$ and 
$t \in \RR$
$$
(e^{-itH_0} \biR e^{itH_0} - e^{-itH} \biR e^{itH}) 
= e^{-itH} ( \Gamma(t) \biR \Gamma(t)^{-1} - \biR )  e^{itH} \in \cC \, . 
$$  
Since the set of polynomials of the resolvents is norm dense 
in $\cR$, it proves that the resolvent algebra is stable under
the perturbed dynamics. 
\qed \end{proof}

\subsection{Inadmissible dynamics} \label{s6.3} 
\setcounter{equation}{0}

We conclude this outline with the remark that the non-interacting relativistic 
Hamiltonian $H_m = (P^2 + m^2)^{1/2}$, $m > 0$,
which is related to the boundary case
$\kappa = 1$ in the preceding lemma, does \textit{not} lead to an automorphism
group of $\cR$. In order to prove this, consider the resolvent of the
position operator $(i c + Q)^{-1}$, $c \in \RR \backslash \{0 \}$. 
By an elementary computation one obtains  
\begin{equation} \label{evolution}
e^{it H_m} (i c + Q)^{-1} e^{-it H_m}
= (i c + Q + tP (P^2 + m^2)^{-1/2})^{-1} \, , \quad t \in \RR \, .
\end{equation}
To see that the operator on the right hand side is not an element 
of the resolvent algebra if $t \neq 0$, 
we consider the unitary representation of the dilations 
$D(\delta)$, $\delta \in \RR_+$, on the underlying
Hilbert space whose adjoint action on the basic resolvents is given by 
$$
D(\delta) (i c + a P + b \, Q)^{-1} D(\delta)^{-1} =
(i c + \delta a P + \delta^{-1} b \, Q)^{-1} \, , 
\quad \delta \in \RR_+ \, . 
$$ 
By arguments already used in the proof of 
\cite[Thm.\ 4.8]{BuGr2}, one obtains the following result.

\begin{lemma} \label{l6.4}
For given $a,b \in \RR$, $c \in \RR \backslash \{ 0 \}$,  
one has in the strong operator topology 
$$\lim_{\delta \rightarrow \infty} \, D(\delta) (i c + a P + b \, Q)^{-1} 
D(\delta)^{-1} = 
\begin{cases}
(i c)^{-1} \, 1  & \mbox{if} \quad a=0 \\
0                     & \mbox{if} \quad a \neq 0 \, .
\end{cases}
$$
Hence 
$\lim_{\delta \rightarrow \infty} \, D(\delta) R D(\delta)^{-1} \in \CC \, 1$
for any $R \in \cR$. 
\end{lemma}
\begin{proof}
By a routine computation, one obtains on the domain $\cD$ the equality 
\begin{align*}
& \big( (i c + \delta a P + \delta^{-1} b \, Q)^{-1} 
- (i c + \delta a P)^{-1} \big) \\ 
& = (i c + \delta a P + \delta^{-1} b \, Q)^{-1} \, (- \delta^{-1} b \, Q) \, 
   (i c + \delta a P)^{-1}  \\ 
& = (i c +  \delta a P  + \delta^{-1} b \, Q)^{-1} 
(i c + \delta a P)^{-1} \big( iab \ (i c + \delta a P)^{-1} 
- \delta^{-1} b \, Q  \big) \, .
\end{align*}
Since the resolvents are uniformly bounded for fixed $c \neq 0$ and
$(i c + \delta a P)^{-1} \rightarrow 0$ in the strong operator
topology if $\delta \rightarrow \infty$ and $a \neq 0$, the first part 
of the statement follows. The second part is then a consequence of the fact
that the polynomials of all resolvents form a norm dense set in $\cR$. 
\qed \end{proof}

The fact that the operator on the right
hand side of equation~(\ref{evolution}) does not belong to $\cR$ 
if $t \neq 0$ is now a consequence of the following lemma. 
\begin{lemma} \label{l6.5}
Let $c \neq 0$. Then
$$
\lim_{\delta \rightarrow \infty} \, D(\delta) 
(i c + Q + t P (P^2 + m^2)^{-1/2})^{-1} 
D(\delta)^{-1} = (i c + t P / |P| )^{-1} 
$$
in the strong operator topology. 
\end{lemma}
\begin{proof}
The adjoint action of the dilations on the operator 
is given by, $\delta \in \RR_+$, 
$$  D(\delta) 
(i c + Q + t P (P^2 + m^2)^{-1/2})^{-1} 
D(\delta)^{-1} = (i c + \delta^{-1} Q + 
t P (P^2 + \delta^{-2} m^2)^{-1/2})^{-1} \, .
$$
Now the operator $(i c + t P/|P|)^{-1}$ maps 
the dense set $\cD_0 \subset \cD$ of test functions, 
vanishing in momentum space in a neighborhood of the origin,
onto itself. The statement then follows from the above equality by a similar
computation as in the preceding argument.
\qed \end{proof}

According to the preceding Lemma \ref{l6.5}, 
the scaling limits of the operators  
\eqref{evolution} are not 
multiples of the identity, so by Lemma \ref{l6.4} 
they are not elements 
of the resolvent algebra.
So we conclude that the resolvent algebra $\cR$ is not stable under
the adjoint action of the unitaries $U_m(t)$ for $t \in \RR \backslash \{0\}$.
But, whereas the relativistic dynamics of particles 
does not leave the quantum mechanical resolvent algebra invariant, 
it is noteworthy that 
there do not appear such problems for the resolvent algebra of the 
corresponding relativistic field theory. 

\section{Conclusions}
\setcounter{equation}{0}

In the present article we have continued our study of the resolvent algebra,
which is a C*-algebraic framework for the description of the kinematics of 
quantum systems, putting emphasis here on its applications. 

It was already pointed out in \cite{BuGr2} that any kinematical algebra  
which is capable of describing a variety of 
dynamics of physical interest, must have ideals. Indeed, the
familiar algebra $\cB(\cH)$ of bounded operators on some Hilbert
space does comply with this condition,
it contains the ideal of compact operators. But, in contrast to the 
resolvent algebra, $\cB(\cH)$ does not contain specific information 
about the underlying quantum system. The ideal structure of 
the resolvent algebra is more complex than that of  $\cB(\cH)$.
Compact operators appear in disguise also in subalgebras of 
the resolvent algebra, \ie they 
are homomorphic to compact operators, but have in general infinite multiplicity.
In fact, one can extract from the nesting of these subalgebras the number 
of degrees of freedom of the system \cite{Bu1}. 

In the present study of infinite lattice systems we have exhibited 
several reasons why this ideal structure is essential. 
First, it greatly simplifies proofs that a given dynamics
acts by automorphisms on the resolvent algebra, cf.\ in this context
the Lemmas \ref{l6.2} and \ref{l6.3}. Second, we made use of the 
fact that non-simple algebras, \ie algebras having ideals, 
in general admit outer bounded 
generators describing perturbations of the dynamics, cf.\ Lemma
\ref{l2.2}. Third, the ideal structure was vital in the description
of finitely localized systems on the global algebra; for there exist
states which are annihilated by parts of its ideals, describing a 
situation, where the quantum system is confined to a box, 
cf.\ Lemma~\ref{l4.3}. 

Making use of these features of the resolvent algebra, we were able to 
establish the existence of global dynamics, of global 
equilibrium states at high temperatures and of global ground states.
All of these states are regular.
Proofs to that effect were not available, to the best of our knowledge. 

Our arguments were based on well known results in the context of C*-dynamical 
systems, the vital ingredient being the fact that we were able to show that 
they are applicable here. The resolvent algebra of the lattice theory 
does not belong to this class; yet it contains a C*-algebra formed by
(in regular representations) weak-*-dense subalgebras, on which the 
dynamics acts pointwise norm-continuously. Exploiting this 
structure, we could confine our analysis to states on this subalgebra and 
extend the locally normal states back to regular states on the full algebra 
at the very end of our analysis. Thus the ideal structure entered
again at this point.

The present results suggest to explore the applicability of the 
resolvent algebra to other problems of physical interest, such 
as Pauli-Fierz type models 
of small systems coupled to a thermal reservoir, cf.\ for example \cite{Sch}; 
there we expect no major problems. More ambitiously, one can also
study within this framework the asymptotic 
time behavior of states, where
one couples two thermal reservoirs, filling half spaces, through 
a finite 
layer and ask, whether they approach a steady state in the course of
time, cf.\ for example \cite{Ru}. It seems also possible to treat 
along these lines models of quantum crystals, where the interaction
between neighboring atoms is described by harmonic forces, cf.\ for
example \cite{MiVeZa}; note that these forces induce automorphisms 
of the resolvent algebra. 

Considerably more involved are 
infinite systems, where all particles can 
interact with each other. There the statistics of particles will enter,
\ie particles can no longer be treated as distinguishable, such as in 
the present context, where each particle is confined 
by binding forces to a neighborhood
of its respective lattice point. As already indicated in connection with 
the no-go theorem 
concerning the relativistic dynamics, presented in Subsection~\ref{s6.3}, 
one should then no longer deal with the particle picture, but rely 
on the field theoretic version of the resolvent algebra \cite{BuGr2}.
We hope to return to these problems elsewhere.

\begin{acknowledgement} 
I am grateful to Hendrik Grundling for numerous discussions,  
valuable hints which helped to clarify my views on the present topic, 
and his critical reading of this article before its publication.  
\end{acknowledgement}

\end{document}